\documentclass[12pt]{article}
\usepackage{amssymb,amsfonts,amsmath,latexsym,amsthm}
\usepackage[usenames,dvipsnames]{color}
\usepackage[]{graphicx}
\usepackage[space]{grffile}
\usepackage{mathrsfs}   % fancy math font
\usepackage{caption}
\usepackage{subcaption}
\usepackage{verbatim}
\usepackage{url}
\usepackage{enumitem}
\usepackage[colorlinks,citecolor=black,linkcolor=black,urlcolor=black]{hyperref}
\usepackage[authoryear,round]{natbib}
\usepackage[ruled,vlined]{algorithm2e}

\theoremstyle{plain}
\newtheorem{theorem}{Theorem}[section]

\theoremstyle{definition}
\numberwithin{equation}{section}

\allowdisplaybreaks

\newcommand{\R}{\mathbb{R}}

\newcommand{\E}{\mathbb{E}}

\newcommand{\blind}{0}

\begin{document}

\def\spacingset#1{\renewcommand{\baselinestretch}%
{#1}\small\normalsize} \spacingset{1}

%%%%%%%%%%%%%%%%%%%%%%%%%%%%%%%%%%%%%%%%%%%%%%%%%%%%%%%%%%%%%%%%%%%%%%%%%%%%%%

\if0\blind
{
  \title{\bf Fast and Accurate Approximation of the Full Conditional for Gamma Shape Parameters}
  \author{Jeffrey W. Miller\thanks{Contact: \url{jwmiller@hsph.harvard.edu}. Department of Biostatistics, Harvard School of Public Health, 655 Huntington Ave 1-419, Boston MA, 02115.}\\
    Harvard University, Department of Biostatistics}
  \maketitle
} \fi

\if1\blind
{
  \bigskip
  \bigskip
  \bigskip
  \begin{center}
    {\LARGE\bf Fast and Accurate Approximation of the Full Conditional for Gamma Shape Parameters}
\end{center}
  \medskip
} \fi

\bigskip
\begin{abstract}
The gamma distribution arises frequently in Bayesian models, but there is not an easy-to-use conjugate prior for the shape parameter of a gamma. This inconvenience is usually dealt with by using either Metropolis--Hastings moves, rejection sampling methods, or numerical integration. However, in models with a large number of shape parameters, these existing methods are slower or more complicated than one would like, making them burdensome in practice. It turns out that the full conditional distribution of the gamma shape parameter is well approximated by a gamma distribution, even for small sample sizes, when the prior on the shape parameter is also a gamma distribution. This article introduces a quick and easy algorithm for finding a gamma distribution that approximates the full conditional distribution of the shape parameter. We empirically demonstrate the speed and accuracy of the approximation across a wide range of conditions. If exactness is required, the approximation can be used as a proposal distribution for Metropolis--Hastings. 
\end{abstract}
%Supplementary material for this article is available online.

\noindent%
{\it Keywords:}  Bayesian, Generalized Newton, Hierarchical models, Markov chain Monte Carlo, Sampling.
\vfill

\newpage
\spacingset{1.45} % DON'T change the spacing!

\section{Introduction}

The lack of a nice conjugate prior for the shape parameter of the gamma distribution is a commonly occurring nuisance in many Bayesian models. Conjugate priors for the shape parameter do exist, but unfortunately, they are not analytically tractable \citep{damsleth1975conjugate,miller1980bayesian}. In Markov chain Monte Carlo (MCMC) algorithms, one can easily use Metropolis--Hastings (MH) moves to update the shape parameter, but mixing can be slow if the proposal distribution is not well-calibrated. Rejection sampling schemes can also be used for sampling the shape parameter \citep{son2006bayesian,pradhan2011bayes}, however, the more efficient schemes such as adaptive rejection sampling \citep{gilks1992adaptive} are complicated and require log-concavity, which does not always hold under some important choices of prior. Modern applications involve models with a large number of parameters, necessitating fast inference algorithms that work very generally with minimal tuning. Indeed, our motivation for developing the approach in this article is a gene expression model involving tens of thousands of gamma shape parameters.

In this article, we introduce a fast and simple algorithm for finding a gamma distribution that approximates the full conditional distribution of the gamma shape parameter, when the prior on the shape parameter is also a gamma distribution.
This algorithm can be used to perform approximate Gibbs updates by sampling from the approximating gamma distribution.
Alternatively, the approximation can be used as an MH proposal distribution to make a move that exactly preserves the full conditional and has high acceptance rate.

The basic idea of the algorithm is to approximate the full conditional density $f$ by a gamma density $g$ chosen such that the first and second derivatives of $\log g$ match those of $\log f$ at a point near the mean of $f$.  Since the mean of $f$ is not known in closed form, the approximation is iteratively refined by matching derivatives at the mean of the current $g$.
% To find a point near the mean of $f$, the approximation is iteratively refined by updating the point at which the derivatives are matched by setting it to the mean of the current $g$.

The article is organized as follows. Section~\ref{section:algorithm} describes the algorithm, and Section~\ref{section:results} contains an empirical assessment of the accuracy and speed of the algorithm. Section~\ref{section:previous} discusses previous work, and Section~\ref{section:derivation} provides mathematical details on the derivation of the algorithm. The supplementary material contains a characterization of the fixed points of the algorithm, additional derivations, and additional empirical results.

\section{Algorithm}
\label{section:algorithm}

Consider the model $X_1,\ldots,X_n|a,\mu \sim \mathrm{Gamma}(\mathrm{shape}=a,\,\mathrm{rate}=a/\mu)$ i.i.d., where $a,\mu>0$, and note that this makes $\mu = \E(X_i\mid a,\mu)$. This parametrization is convenient, since $a$ controls the concentration and $\mu$ controls the mean. 
We assume a gamma prior on the shape, say, $a \sim \mathrm{Gamma}(\mathrm{shape}=a_0,\, \mathrm{rate}=b_0)$ where, if desired, $a_0$ and $b_0$ can depend on~$\mu$.
The following algorithm produces $A$ and $B$ such that
$$p(a\mid x_1,\ldots,x_n,\mu,a_0,b_0)\approx \mathrm{Gamma}(a\mid\mathrm{shape}=A,\,\mathrm{rate}=B).$$
We use $\psi(x)$ and $\psi'(x)$ to denote the digamma and trigamma functions, respectively, and $\log(x)$ denotes the natural log.

\begin{algorithm}
\SetKwInOut{Input}{input}\SetKwInOut{Output}{output}
\DontPrintSemicolon
\Input{data $x_1,\ldots,x_n > 0$, parameters $\mu,a_0,b_0 > 0$, tolerance $\epsilon > 0$, and maximum number of iterations $M$.}
\Output{$A$ and $B$.}
\Begin{
    $R \leftarrow \sum_{i = 1}^n \log(x_i)$, $S \leftarrow \sum_{i = 1}^n x_i$, and $T \leftarrow S/\mu - R + n\log(\mu) - n$\;
    % $T \leftarrow (1/\mu)\sum_{i = 1}^n x_i - \sum_{i = 1}^n \log(x_i) + n\log(\mu) - n$\;
    $A \leftarrow a_0 + n/2$ and $B \leftarrow b_0 + T$\;
    \For{$j = 1,\ldots,M$}{
        $a \leftarrow A/B$\;
        $A \leftarrow a_0 - n a + n a^2 \psi'(a)$\;
        $B \leftarrow b_0 + (A-a_0)/a - n\log(a) + n\psi(a) + T  $\;
        \lIf{$|a/(A/B) - 1| < \epsilon$}{\KwRet{$A,B$}}
    }
    \KwRet{$A,B$}
}
\caption{Approximating the full conditional of the shape parameter\label{algorithm:main}}
\end{algorithm}

We recommend setting $\epsilon = 10^{-8}$ and $M = 10$.  Using $\epsilon = 10^{-8}$, the algorithm terminates in four iterations or less in all of the simulations we have done, so $M = 10$ is conservative.  
Most languages have routines for the digamma and trigamma functions, which are defined by $\psi(x) = \frac{\partial}{\partial x}\log\Gamma(x)$ and $\psi'(x) = \frac{\partial}{\partial x}\psi(x)$, where $\Gamma(x)$ is the gamma function.
See Section~\ref{section:derivation} for the derivation of the algorithm. 
See the supplement for a fixed point analysis.

To perform an approximate Gibbs update to $a$ in an MCMC sampler, run Algorithm~\ref{algorithm:main} to obtain $A,B$ using the current values of $x_{1:n},\mu,a_0,b_0$, and then sample $a\sim\textrm{Gamma}(A,B)$. (For brevity, in the rest of the paper, we write $x_{1:n}$ to denote $(x_1,\ldots,x_n)$, and we write $\textrm{Gamma}(\alpha,\beta)$ to denote the gamma distribution with shape $\alpha$ and rate $\beta$.)

Alternatively, to perform a Metropolis--Hastings move to update $a$, one can run Algorithm~\ref{algorithm:main} to obtain $A,B$ using the current values of $x_{1:n},\mu,a_0,b_0$, then sample a proposal $a'\sim\textrm{Gamma}(A,B)$, and accept with probability 
$$\min\Big\{1,\; \frac{p(a'\mid x_{1:n},\mu,a_0,b_0) \textrm{Gamma}(a|A,B) }{p(a\mid x_{1:n},\mu,a_0,b_0) \textrm{Gamma}(a'|A,B)}\Big \},$$
where $a$ is the current value of the shape parameter. Note that the normalization constant of $p(a\mid x_{1:n},\mu,a_0,b_0)$ does not need to be known since it cancels.

To infer the mean $\mu$, note that the inverse-gamma family is a conditionally conjugate prior for $\mu$ given $a$. Therefore, Gibbs updates to $\mu$ can easily be performed when using an inverse-gamma prior on $\mu$, independently of $a$, given $a_0$ and $b_0$.

\section{Performance assessment}
\label{section:results}

To assess the performance of Algorithm~\ref{algorithm:main}, we evaluated its accuracy in approximating the true full conditional distribution, and its speed in terms of the number of iterations until convergence, across a wide range of simulated conditions.
%The accuracy and speed are surprisingly good across all cases that we tried.

Specifically, for each combination of $n \in \{1,10,100\}$, $r \in \{0.5, 1, 2 \}$, $a_\mathrm{true} \in \{10^{-6},10^{-5},\ldots,10^5,10^6 \}$, $\mu_\mathrm{true} \in \{10^{-6},10^{-5},\ldots,10^5,10^6 \}$, and $a_0 \in \{1, 0.1, 0.01\}$, we generated five independent data sets $x_1,\ldots,x_n \sim \mathrm{Gamma}(a_\mathrm{true},a_\mathrm{true}/\mu_\mathrm{true})$ i.i.d., ran Algorithm~\ref{algorithm:main} with inputs $x_{1:n}$, $\mu = r \mu_\mathrm{true}$, $a_0$, and $b_0=a_0$, and then compared the resulting $\mathrm{Gamma}(a|A,B)$ approximation to the true full conditional distribution $p(a \mid x_{1:n},\mu,a_0,b_0)$.  For the prior parameters, we set $b_0 = a_0$ in each case so that the prior mean is $1$.  In each case, we used $\epsilon = 10^{-8}$ for the convergence tolerance, and $M = 10$ for the maximum number of iterations. 

\subsection{Approximation accuracy}

First, to get a visual sense of the closeness of the approximation, Figure~\ref{figure:CDFs} shows the cumulative distribution functions (CDFs) of the true and approximate full conditional distributions for a single simulated data set for each case with $r = 1$, $a_\mathrm{true} = 1$, and $\mu_\mathrm{true} = 1$.  In each case, the approximation is close to the true distribution, and in most cases the approximation is visually indistinguishable from the truth.  These plots are representative of the full range of cases considered.
% --- similar plots were observed across the other cases.

\begin{figure}%[htbp]
  \centering
  \includegraphics[trim=.3cm 0 1.2cm 0, clip, width=0.325\textwidth]{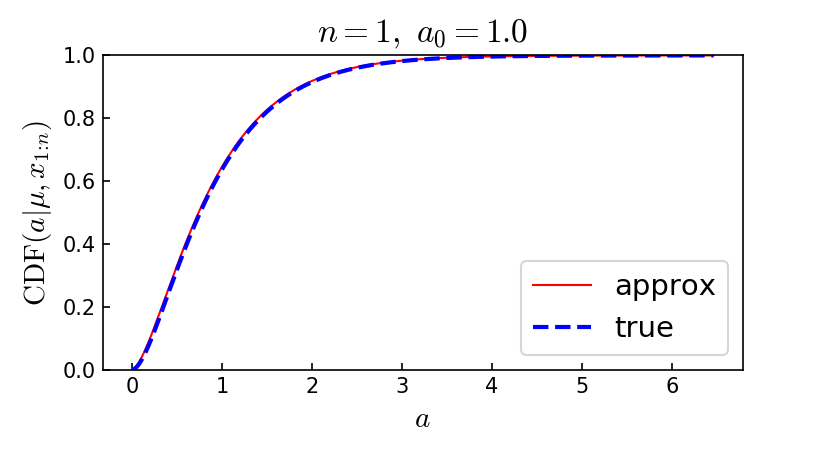}
  \includegraphics[trim=.3cm 0 1.2cm 0, clip, width=0.325\textwidth]{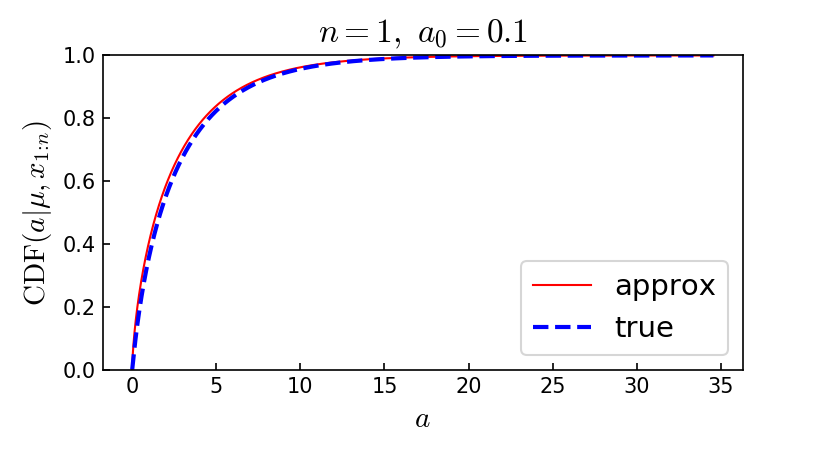}
  \includegraphics[trim=.3cm 0 1.2cm 0, clip, width=0.325\textwidth]{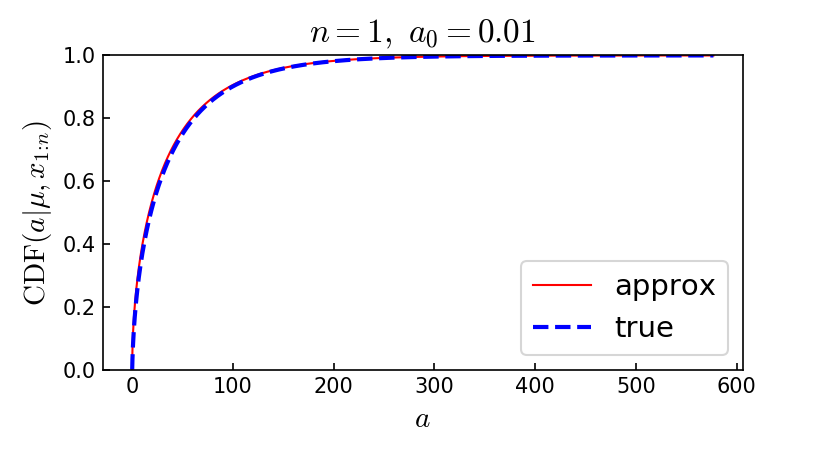}
  \includegraphics[trim=.3cm 0 1.2cm 0, clip, width=0.325\textwidth]{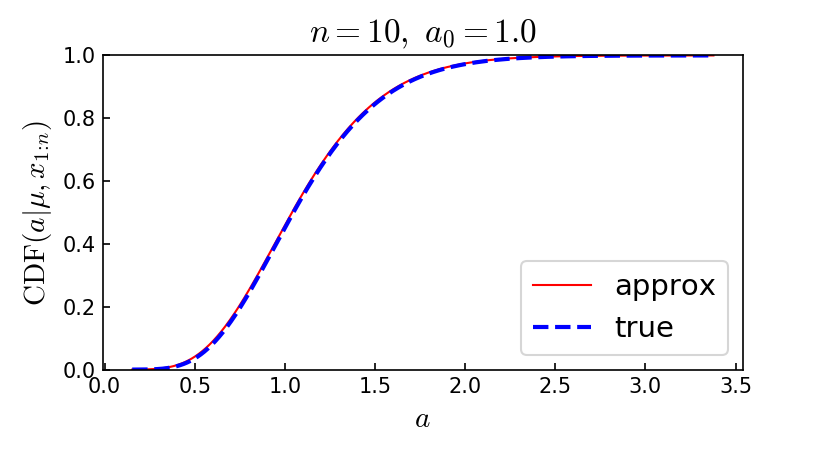}
  \includegraphics[trim=.3cm 0 1.2cm 0, clip, width=0.325\textwidth]{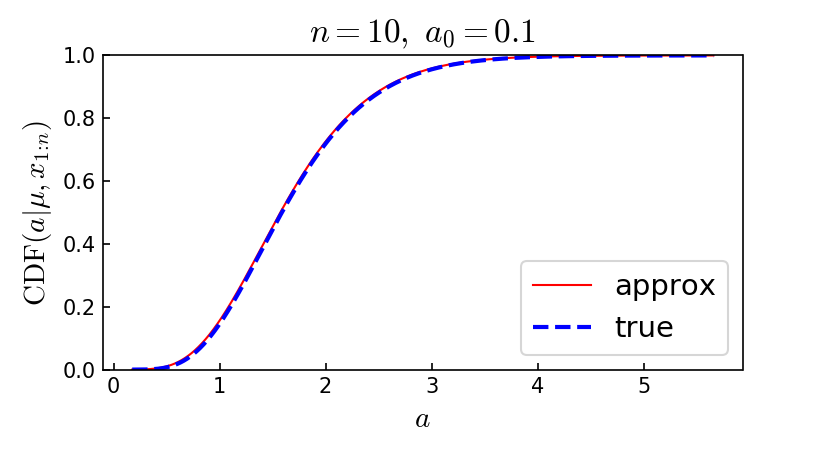}
  \includegraphics[trim=.3cm 0 1.2cm 0, clip, width=0.325\textwidth]{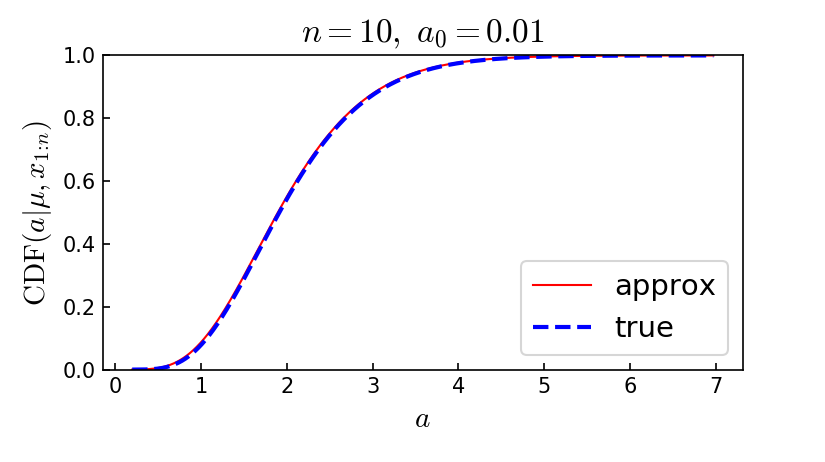}
  \includegraphics[trim=.3cm 0 1.2cm 0, clip, width=0.325\textwidth]{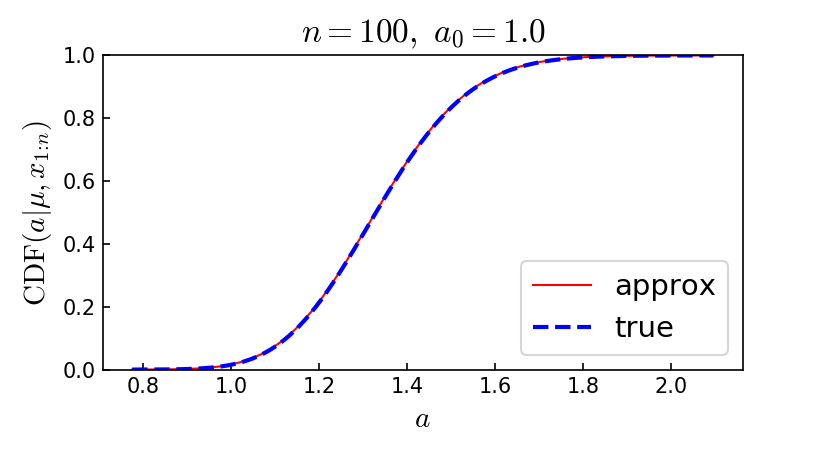}
  \includegraphics[trim=.3cm 0 1.2cm 0, clip, width=0.325\textwidth]{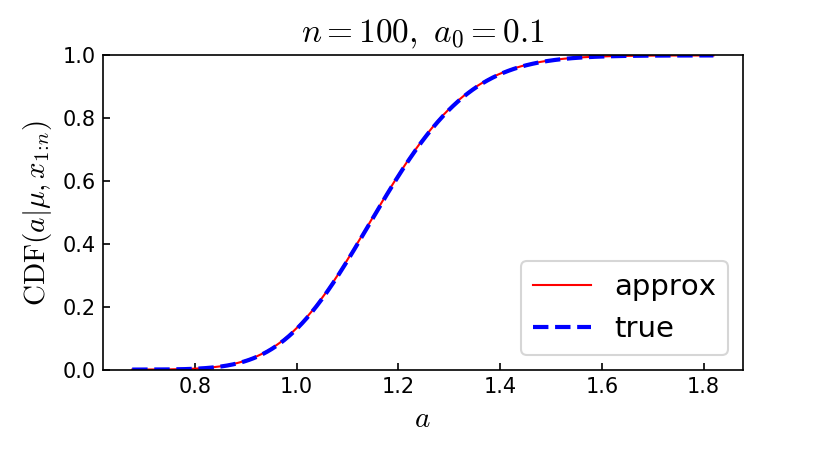}
  \includegraphics[trim=.3cm 0 1.2cm 0, clip, width=0.325\textwidth]{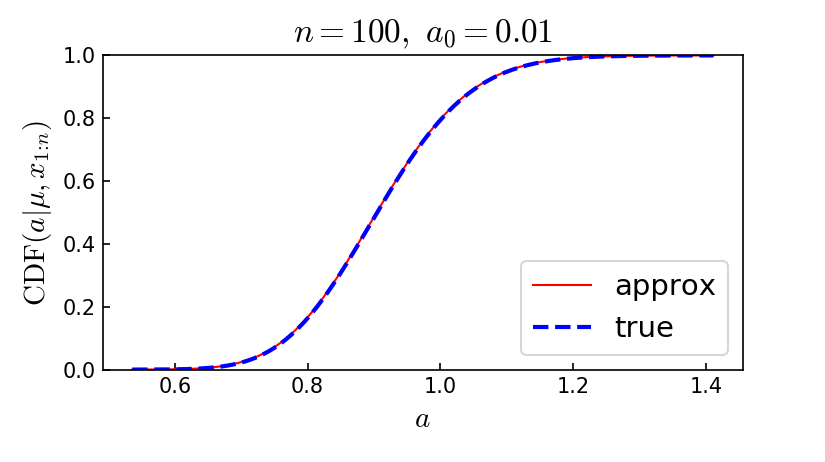}
  \caption{Cumulative distribution function (CDF) of the true full conditional $p(a\mid x_{1:n},\mu,a_0,b_0)$ and approximate full conditional, on simulated data $x_1,\ldots,x_n \sim \mathrm{Gamma}(1,1)$, when conditioning on $\mu = 1$ and the prior is $a\sim\mathrm{Gamma}(a_0,a_0)$.}
  \label{figure:CDFs}
\end{figure}

\begin{figure}%[htbp]
  \centering
  \includegraphics[trim=.5cm 0 1cm 0, clip, width=0.325\textwidth]{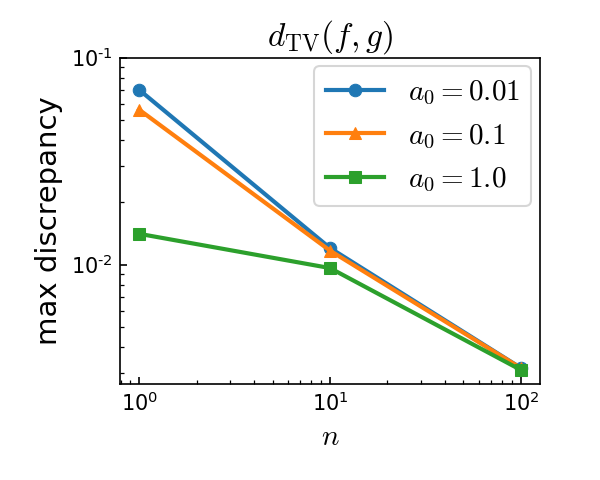}
  \includegraphics[trim=.5cm 0 1cm 0, clip, width=0.325\textwidth]{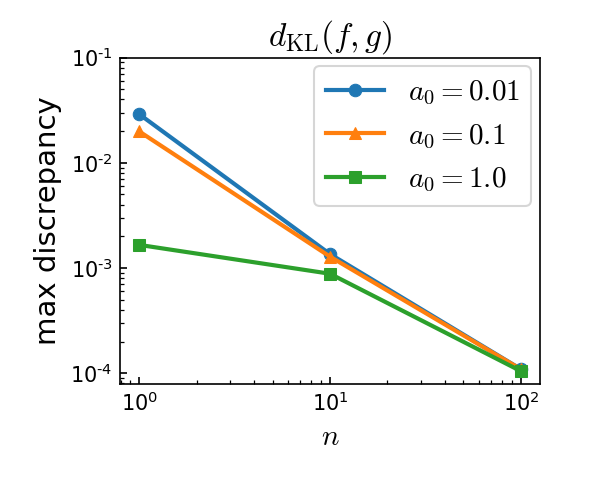}
  \includegraphics[trim=.5cm 0 1cm 0, clip, width=0.325\textwidth]{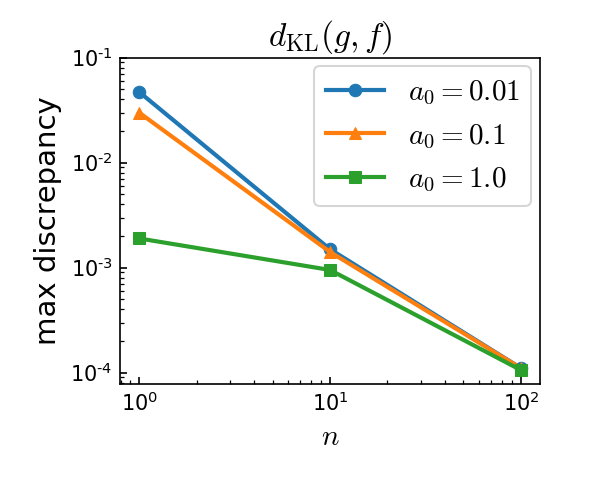}
  \caption{Largest observed discrepancy between the true full conditional $f$ and the approximate full conditional $g$, for the total variation distance $d_\mathrm{TV}(f,g)$ and the Kullback--Leibler divergences $d_\mathrm{KL}(f,g)$ and $d_\mathrm{KL}(g,f)$, for each $a_0$ and $n$.}
  \label{figure:worst}
\end{figure}

To quantify the discrepancy between the truth and the approximation, for each simulation run we computed:
\begin{enumerate}
\item the total variation distance, $d_\mathrm{TV}(f,g) = \frac{1}{2}\int |f(a)-g(a)|d a$,
\item the Kullback--Leibler divergence, $d_\mathrm{KL}(f,g) = \int f(a)\log(f(a)/g(a)) d a$, and
\item the reverse Kullback--Leibler divergence, $d_\mathrm{KL}(g,f) = \int g(a)\log(g(a)/f(a)) d a$, 
\end{enumerate}
% (i) the total variation distance, $d_\mathrm{TV}(f,g) = \frac{1}{2}\int |f(a)-g(a)|d a$,
% (ii) the Kullback--Leibler divergence, $d_\mathrm{KL}(f,g) = \int f(a)\log(f(a)/g(a)) d a$, and
% (iii) the Kullback--Leibler divergence, $d_\mathrm{KL}(g,f) = \int g(a)\log(g(a)/f(a)) d a$, 
where $f(a) = p(a \mid x_{1:n},\mu,a_0,b_0)$ is the true density and $g(a) = \mathrm{Gamma}(a|A,B)$ is the approximate density.  Numerical integration was performed to compute these integrals in order to compare with the truth; see the supplementary material for details.

Figure~\ref{figure:worst} shows the largest observed value (that is, the worst-case value) of $d_\mathrm{TV}(f,g)$, $d_\mathrm{KL}(f,g)$, and $d_\mathrm{KL}(g,f)$ across all cases for each $a_0$ and $n$, using the average discrepancy over the five data sets for each case.  This shows that the approximation is quite good across a very wide range of cases.  The approximation accuracy improves as $n$ increases, and the accuracy is decent even when $n$ is small. 
The accuracy improves as $a_0$ increases, presumably because this makes the prior stronger and the prior is itself a gamma distribution, so it makes sense that the fit of a gamma approximation would improve.

For a more fine-grained breakdown of the discrepancies, case-by-case, see the supplementary material for heatmaps showing
$d_\mathrm{TV}(f,g)$, $d_\mathrm{KL}(f,g)$, and $d_\mathrm{KL}(g,f)$ across a range of cases.
The accuracy tends to be worst in the cases with $a_\mathrm{true}=1$, except that sometimes when $n=1$ it is worst when $a_\mathrm{true}$ is larger.
The accuracy tends to level off as $a_\mathrm{true}$ grows larger than $100$ or smaller than $0.01$.
Across all cases, the approximation is particularly good when $a_\mathrm{true}$ is $0.01$ or smaller, or when $n$ is large and $\mu=\mu_\mathrm{true}$.
% The accuracy seems unaffected by $\mu_\mathrm{true}$.
% , which is not too surprising due to the scaling property of gamma distributions.

\subsection{Speed of convergence}

To evaluate the speed of Algorithm~\ref{algorithm:main} in terms of the number of iterations required, Table~\ref{table:iterations} shows the number of runs in which the algorithm required $k$ iterations to reach the termination condition, for each $k$. All runs terminated in four iterations or less. (For each $a_0$, these results are over all of the 1521 combinations of $n$, $r$, $a_\mathrm{true}$, $\mu_\mathrm{true}$, times five replicates each, as described above.) The algorithm terminates very rapidly in all cases considered.

\begin{table}  %[b]
\centering
\caption{Number of runs that required $k$ iterations before terminating.}%
\label{table:iterations}%
\begin{tabular}{|c|c|c|c|}%
\hline
$k$ & $a_0 = 1$ & $a_0 = 0.1$ & $a_0 = 0.01$ \\
\hline
1 iteration & 0 & 0 & 0 \\
2 iterations & 0 & 318 & 631 \\
3 iterations & 5751 & 4699 & 4308 \\
4 iterations & 1854 & 2588 & 2666 \\
$\geq$ 5 iterations & 0 & 0 & 0 \\
\hline
\end{tabular}
\end{table}

% Histogram of iterations required (a0 = 0.01):
% 1 iterations were required in 0 of the runs
% 2 iterations were required in 631 of the runs
% 3 iterations were required in 4308 of the runs
% 4 iterations were required in 2666 of the runs

% Histogram of iterations required (a0 = 0.1):
% 1 iterations were required in 0 of the runs
% 2 iterations were required in 318 of the runs
% 3 iterations were required in 4699 of the runs
% 4 iterations were required in 2588 of the runs

% Histogram of iterations required (a0 = 1.0):
% 1 iterations were required in 0 of the runs
% 2 iterations were required in 0 of the runs
% 3 iterations were required in 5751 of the runs
% 4 iterations were required in 1854 of the runs

\section{Previous work}
\label{section:previous}

The original inspiration for Algorithm~\ref{algorithm:main} was the generalized Newton algorithm of \citet{minka2002estimating} for maximum likelihood estimation of gamma parameters. That said, our algorithm is significantly different since we are considering the full conditional distribution of the shape, rather than finding the unconditional MLE of the shape. Further, we iteratively set $a$ to the mean rather than the mode of the approximating density. This is an important difference since the mode can be at $0$ when $a_0 < 1$, but the derivatives blow up at $0$, making it incoherent to match the derivatives at $0$. Meanwhile, our algorithm handles $a_0 < 1$ without difficulties. Also, empirical evidence suggests that even when the mode is valid, using the mean tends to yield a more accurate approximation.

\citet{damsleth1975conjugate} found conjugate priors for the shape $a$ in the case of known rate $b$, as well as the case of unknown $b$. However, it seems that the distributions in these conjugate families are not easy to work with, analytically. Consequently, \citet{damsleth1975conjugate} resorts to numerical integration over $a$, for each value of $b$, in order to obtain $p(b\mid x_{1:n}) = \int p(a,b\mid x_{1:n})\,d a$, the posterior density of $b$.
Similarly, \citet{miller1980bayesian} numerically integrates over $a$ to obtain the posterior of $b$ (or of the mean $\mu$), but instead of doing so for each value of $b$, he numerically computes the mean, variance, skewness, and kurtosis of $p(b\mid x_{1:n})$, and then constructs an approximation to $p(b\mid x_{1:n})$ by using the method of moments to fit a member of the Pearson family to it.  \citet{miller1980bayesian} uses improper priors as well as the conjugate priors of \citet{damsleth1975conjugate}.
These numerical integration approaches would be fine for a small number of shape parameters with fixed data $x_{1:n}$ and fixed hyperparameters $a_0,b_0$. However, modern applications involve hierarchical models with many shape parameters, with latent $x$'s and unknown hyperparameters. In such cases, the numerical integration approaches of \citet{damsleth1975conjugate} and \citet{miller1980bayesian} would be very computationally intensive.

\citet{son2006bayesian} propose to use the adaptive rejection sampling (ARS) method of \citet{gilks1992adaptive} to sample from the full conditional for the shape $a$ in a Gibbs sampler for $a$ and $b$. There are a few disadvantages to this approach. ARS is quite a bit more complicated to implement than our procedure, and requires additional bookkeeping that is burdensome when dealing with a large number of shape parameters. Further, ARS requires the target distribution to be log-concave, but log-concavity of the full conditional of $a$ is not always guaranteed. For instance, when using a $\mathrm{Gamma}(a_0,b_0)$ prior on $a$, if $0 < a_0 < 1$ then the prior is not log-concave, which can cause the full conditional given $\mu,x_{1:n}$ to fail to be log-concave when $a_0$ is sufficiently small and $n=1$. In practice, it is important to allow $a_0$ to be small in order to obtain a less informative prior while holding the prior mean $a_0/b_0$ fixed.
% Note to self: I think the full conditional given b,x may be log-concave even when a_0 < 1.
Rejection sampling (RS) approaches have also been used by other authors, for instance, \citet{pradhan2011bayes} use the log-concave RS method of \citet{devroye1984simple}, and \citet{tsionas2001exact} uses RS with a calibrated Gaussian proposal.
% \citet{tsionas2001exact} uses rejection sampling with a calibrated Gaussian proposal distribution.

% to sample from the posterior $p(a \mid x_{1:n}) = \int p(a,b \mid x_{1:n})$ with the rate $b$ integrated out analytically. This also has some disadvantages. First,  Again, this requires the target distribution to be log-concave (in this case, $p(a \mid x_{1:n})$), 

\section{Derivation of the algorithm}
\label{section:derivation}

Let $f(a) = p(a \mid x_{1:n},\mu,a_0,b_0)$ denote the true full conditional density, and let $g(a) = \mathrm{Gamma}(a|A,B)$ for some $A,B$ to be determined.
The basic idea of the algorithm is to choose $A$ and $B$ so that the first and second derivatives of $\log g$ match those of $\log f$ at a point $a$ near the mean of $f$.  To find a point near the mean of $f$, the algorithm is initialized with a choice of $A$ and $B$ based on an approximation of the gamma function (see the supplementary material), and then iteratively, $A$ and $B$ are refined by matching the first and second derivatives of $\log f$ at the mean $A/B$ from the previous iteration.

It might seem like a better idea to use a point near the mode of $f$ rather than the mean, by analogy with the Laplace approximation. However, on this problem, empirical evidence suggests that using the mean tends to provide better accuracy than the mode. Using the mean also simplifies the algorithm since the mode of $g$ is sometimes at $0$, but the derivatives of $\log f$ and $\log g$ blow up at $0$.
To interpret the limiting behavior of the algorithm, we characterize its fixed points in the supplementary material.

\subsection{Matching derivatives}

First, to find the formulas for matching the first and second derivatives of $\log g$ to those of $\log f$ at some point $a>0$, note that
\begin{align*}
\log g(a) &= A\log(B) - \log \Gamma(A) + (A-1)\log(a) - B a \\
\frac{\partial}{\partial a}\log g(a) &= \frac{A-1}{a} - B \\
\frac{\partial^2}{\partial a^2}\log g(a) &= -\frac{A-1}{a^2}.
\end{align*}
Meanwhile, $f(a) \propto p(x_{1:n} \mid a,\mu,a_0,b_0)\, p(a\mid\mu,a_0,b_0)$ and
\begin{align}\label{eqn:f}
p(x_{1:n} \mid a,\mu,a_0,b_0) &= \prod_{i = 1}^n \frac{(a/\mu)^a}{\Gamma(a)} x_i^{a-1} \exp\!\big(-(a/\mu)x_i\big) \notag \\
&= \frac{(a/\mu)^{n a}}{\Gamma(a)^n} \exp\!\big((a-1)R\big) \exp\!\big(-(a/\mu)S\big) \notag\\
&= \frac{a^{n a}}{\Gamma(a)^n} \exp\!\big(-R - (T+n)a \big) 
\end{align}
where $R =  \sum_{i=1}^n \log(x_i)$, $S = \sum_{i=1}^n x_i$, and $T = S/\mu - R + n\log(\mu) - n$.
Thus,
\begin{align}\label{equation:derivatives}
\log f(a) &= \mathrm{const} +  \log p(x_{1:n} \mid a,\mu,a_0,b_0) + \log p(a\mid \mu,a_0,b_0) \notag\\
&= \mathrm{const} + n a \log(a) - n\log\Gamma(a) - (T+n)a + (a_0-1)\log(a) - b_0 a \notag\\
\frac{\partial}{\partial a}\log f(a) &= n\log(a) + n - n \psi(a) - (T+n) + \frac{a_0-1}{a} - b_0 \\
\frac{\partial^2}{\partial a^2}\log f(a) &= n/a - n\psi'(a) - \frac{a_0-1}{a^2} \notag
\end{align}
since $\psi(a) = \frac{\partial}{\partial a}\log\Gamma(a)$ and $\psi'(a) = \frac{\partial^2}{\partial a^2}\log\Gamma(a)$.
Setting $\frac{\partial^2}{\partial a^2}\log g(a) = \frac{\partial^2}{\partial a^2}\log f(a)$ yields
\begin{equation}\label{eqn:A}
% A = a_0 - a^2 n \big(1/a - \psi'(a)\big),
A = a_0 - n a + n a^2 \psi'(a),
\end{equation}
and setting $\frac{\partial}{\partial a}\log g(a) = \frac{\partial}{\partial a}\log f(a)$ yields
\begin{equation}\label{eqn:B}
% B = b_0 + \frac{A-a_0}{a} - n\big(\log(a) - \psi(a)\big) + T.
B = b_0 + \frac{A-a_0}{a} - n\log(a) + n\psi(a) + T.
\end{equation}
The updates to $A$ and $B$ in Algorithm~\ref{algorithm:main} are defined by Equations~\ref{eqn:A} and \ref{eqn:B}.

\subsection{Updating the matching point $a$}

Given $A$ and $B$ such that $\mathrm{Gamma}(A,B)$ approximates $f$, we can choose $a$ to approximate the mean of $f$ by setting it to the mean of $\mathrm{Gamma}(A,B)$, namely, $a = A/B$.
By alternating between updating $a \leftarrow A/B$, and updating $A,B$ via Equations~\ref{eqn:A} and \ref{eqn:B}, we iteratively refine the choice of $a$.
This yields the loop in Algorithm~\ref{algorithm:main}. 
Convergence is assessed by checking if the relative change in $a$ is smaller than the tolerance $\epsilon$, that is, if $|a_\mathrm{old} - a_\mathrm{new}| < \epsilon a_\mathrm{new}$, or equivalently, $|a_\mathrm{old}/a_\mathrm{new} - 1| < \epsilon$.

% We also explored targeting the median rather than the mean as a point $a$ at which to match derivatives.  Since there is not a closed-form expression for the median of a gamma distribution, we experimented with updating $a \leftarrow (A/B)(3 A - 0.8)/(3 A + 0.2)$, which is an approximation to the median proposed by . This yields a slight increase in accuracy, at the expense of 1-2 additional iterations.

% \if0\blind
% {
% \section{Acknowledgments}
% \label{section:acknowledgments}
% I would like to thank William Townes for helpful conversations.
% } \fi

\bigskip
\begin{center}
{\large\bf SUPPLEMENTARY MATERIAL}
\end{center}

\begin{description}
\item[Gshape.jl:] Source code implementing the algorithm described in the article, as well as the code used to generate the figures in the article. (file type: Julia source code)
\item[Supplement.pdf:] Additional analysis and empirical results. (file type: PDF)
\end{description}

\small
\spacingset{1}
\bibliographystyle{abbrvnat}
% \bibliography{main}

\normalsize
\spacingset{1.6}

%%%%%%%%%%%%%%%%%%%%%%%%%%%%%%%%%%%%%%%%%%%%%%%%%%%%%%%%%%%%%%%%%%%%%%%%%%%%%%%%
% Supplementary material

\newpage
\setcounter{page}{1}
\setcounter{section}{0}
\setcounter{table}{0}
\setcounter{figure}{0}
\renewcommand{\theHsection}{SIsection.\arabic{section}}
\renewcommand{\theHtable}{SItable.\arabic{table}}
\renewcommand{\theHfigure}{SIfigure.\arabic{figure}}
\renewcommand{\thepage}{S\arabic{page}}  
\renewcommand{\thesection}{S\arabic{section}}   
\renewcommand{\thetable}{S\arabic{table}}   
\renewcommand{\thefigure}{S\arabic{figure}}

\bigskip
\bigskip
\bigskip
\begin{center}
{\large\bf Supplementary material for ``Fast and Accurate Approximation of the Full Conditional for Gamma Shape Parameters''}\\
\textit{Jeffrey W. Miller}
\end{center}
\medskip

\section{Fixed points of the algorithm}
\label{section:fixedpoints}

% Note to self:  This section has been thoroughly checked.  JWM 7/28/2018

In this section, we analyze the limiting behavior of Algorithm~\ref{algorithm:main} by characterizing its fixed points.
This provides a direct interpretation of the final point $a$ at which we match the first and second derivatives of $\log g$ to those of $\log f$
(in the notation of Section~\ref{section:derivation}) in the last iteration of the algorithm.  
% to construct the approximation to the full conditional.
% This provides a direct interpretation of the choice of $a$ that is implicitly defined by the algorithm,
% where $a$ is the final point at which we match the derivatives of $\log g$ to those of $\log f$,
% in the notation of Section~\ref{section:derivation}.

In Theorem~\ref{theorem:fixedpoint}, we show that $a>0$ is a fixed point of Algorithm~\ref{algorithm:main} if and only if $\frac{\partial}{\partial a} \log f(a) + 1/a = 0$, and further, the algorithm is guaranteed to have a fixed point.  
It might seem more natural to instead match derivatives at a point where $\frac{\partial}{\partial a} \log f(a) = 0$, however, empirical evidence indicates that this tends to yield a worse approximation.
% Note to self:
% If n>0, it looks like there is probably always a point such that d_1(a)=0.  From a slightly modified version of the proof below, I think 
% that using the mode to update a, i.e., a <- (A-1)/B, has fixed points where d_1(a)=0.
% Older version (let's not use this, due to the preceding observation):
% ... it can happen that there is no such point.
% As a trivial example, suppose $n=0$ and $0 < a_0 < 1$, so that solving $\frac{\partial}{\partial a} \log f(a) = 0$ would require that $(a_0-1)/a - b_0 = 0$, i.e., $a = (a_0-1)/b_0$, but this would make $a<0$ which is invalid.

To interpret this result, note that if $f(a)$ were exactly equal to a gamma density, then the mean of $f$ would be the unique point $a$ such that 
$\frac{\partial}{\partial a} \log f(a) + 1/a = 0$.
There is also a natural interpretation of $\frac{\partial}{\partial a} \log f(a) + 1/a = 0$ in terms of optimization with a logarithmic barrier.
Barrier functions are frequently used in optimization to deal with inequality constraints.  Adding the logarithmic barrier function $\log a$ to the objective function $\log f(a)$ allows one to smoothly enforce the constraint that $a>0$.  This leads to the penalized objective function $\log f(a) + \log a$, and it turns out that the fixed points of Algorithm~\ref{algorithm:main} coincide with the points where the derivative of this penalized objective equals $0$. 
% $\frac{\partial}{\partial a} \log f(a) + 1/a = 0$.

\begin{theorem}
\label{theorem:fixedpoint}
Let $x_1,\ldots,x_n > 0$, where $n\geq 0$, and let $\mu,a_0,b_0 > 0$.  For any $a>0$, $\frac{\partial}{\partial a} \log f(a) + 1/a = 0$ if and only if $a$ is a fixed point of Algorithm~\ref{algorithm:main}.  Further, there exists $a > 0$ such that $\frac{\partial}{\partial a} \log f(a) + 1/a = 0$.  
\end{theorem}
\begin{proof}
First, we show that $\frac{\partial}{\partial a} \log f(a) + 1/a = 0$ characterizes the fixed points.
Define $d_1(a) = \frac{\partial}{\partial a} \log f(a)$ and $d_2(a) = \frac{\partial^2}{\partial a^2} \log f(a)$.
Recall from Section~\ref{section:derivation} that at each iteration, the algorithm updates $a$ by setting $a\leftarrow A/B$ where $A$ and $B$ are defined such that $d_1(a) = \frac{\partial}{\partial a} \log g(a) = (A-1)/a - B$ and $d_2(a) = \frac{\partial^2}{\partial a^2} \log g(a) = -(A-1)/a^2$ for the current value of $a$.
Solving for $A$ and $B$ implies that $A = -a^2 d_2(a) + 1$ and $B = -a d_2(a) - d_1(a)$.
Thus, $a$ is a fixed point of the algorithm precisely when
$a = A/B = (a^2 d_2(a) - 1) / (a d_2(a) + d_1(a))$, or equivalently, by rearranging and canceling terms, when $d_1(a) + 1/a = 0$.

To justify rearranging and cancelling, note that $d_1(a)$ and $d_2(a)$ are finite for all $a\in(0,\infty)$; see Equation~\ref{equation:derivatives}. If $a$ is a fixed point then $a d_2(a) + d_1(a) \neq 0$ since otherwise $A/B$ is infinite or undefined. On the other hand, if $d_1(a) + 1/a = 0$ then $a d_2(a) + d_1(a) \neq 0$, since otherwise $1/a^2 = d_2(a) = n/a - n\psi'(a) - \frac{a_0-1}{a^2}$, which would imply that $n\big(\psi'(a) - 1/a\big) = -a_0/a^2 < 0$, but this is a contradiction since $\psi'(a) - 1/a > 0$ for all $a>0$ \citep[Eqn.\ 13]{qi2009refinements}.

Now, we show that there always exists $a>0$ such that $d_1(a) + 1/a = 0$.
By Equation~\ref{equation:derivatives}, 
\begin{align*} %\label{equation:d11a}
d_1(a) + 1/a &= n\log(a) + n - n \psi(a) - (T+n) + a_0/a - b_0 \\
&= n\big(\log(a) - \psi(a)\big) + a_0/a - b_0 - T.
\end{align*}
Below, we show that $T\geq 0$. Assuming $T\geq 0$ for the moment, it follows from the Intermediate Value Theorem that $d_1(a) + 1/a = 0$ for some $a>0$, since
$\log(a) - \psi(a) \to \infty$ as $a\to 0$ and $\log(a) - \psi(a) \to 0$ as $a\to \infty$ \citep[Eqn.\ 2.2]{alzer1997some}.

To verify that $T\geq 0$, recall from Equation~\ref{equation:derivatives} that $T = S/\mu - R + n\log(\mu) - n$, where $S = \sum_{i=1}^n x_i$ and $R = \sum_{i=1}^n \log(x_i)$.
By rearranging terms, $T = \sum_{i=1}^n \big(x_i/\mu - \log(x_i/\mu)- 1\big)$, and by straightforward calculus, one can check that $x - \log(x) - 1 \geq 0$ for all $x > 0$.

% If we can show that (a) $T \geq 0$, (b) $\log(a) - \psi(a) \to \infty$ as $a\to 0$, and (c) $\log(a) - \psi(a) \to 0$ as $a\to \infty$.
% then by the Intermediate Value Theorem, it will follow that $d_1(a) + 1/a = 0$ for some $a>0$.

% (a) Recall from Equation~\ref{equation:derivatives} that $T = S/\mu - R + n\log(\mu) - n$, where $S = \sum_{i=1}^n x_i$ and $R = \sum_{i=1}^n \log(x_i)$.
% By rearranging terms, $T = \sum_{i=1}^n \big(x_i/\mu - \log(x_i/\mu)- 1\big)$.  By straightforward calculus, one can check that $x - \log(x) - 1 \geq 0$ for all $x > 0$, and therefore, $T \geq 0$.

% (b) The digamma function satisfies the recurrence relation, $\psi(a+1) = \psi(a) + 1/a$ \citep{bernardo1976algorithm}, 
% and thus, $\psi(a) = \psi(a+1) - 1/a$.  Hence, 
% $$ \frac{\log(a) - \psi(a)}{1/a} = a\log(a) - a\psi(a+1) + 1 \to 1$$
% as $a \to 0$, and therefore, $\log(a) - \psi(a) \sim 1/a$ as $a\to 0$.  In particular, $\log(a) - \psi(a) \to \infty$ as $a \to 0$.

% (c) By \citet{bernardo1976algorithm}, $\psi(a) = \log(a) - 1/(2a) + O(1/a^2)$ as $a \to \infty$.  Hence, $\log(a) - \psi(a) \to 0$ as $a \to \infty$.

\end{proof}

\section{Initializing via approximation of $\Gamma(a)$}
\label{section:initialization}

To find good starting values of $A$ and $B$, we consider approximations to the gamma function $\Gamma(a)$ in two regimes: (i) when $a$ is large, and (ii) when $a$ is small.

By Stirling's approximation to $\Gamma(a)$, we have that $\Gamma(a)/a^a \sim \sqrt{2\pi}\,a^{-1/2} e^{-a}$ as $a\to \infty$. 
(Here, $h_1(a)\sim h_2(a)$ as $a\to a^*$ means that $h_1(a)/h_2(a)\to 1$ as $a\to a^*$.)
Plugging this into Equation~\ref{eqn:f}, we get
$$p(x_{1:n} \mid a,\mu,a_0,b_0) \approx 
(2\pi)^{-n/2} a^{n/2} e^{n a} \exp\!\big(-R - (T+n)a \big) \propto a^{n/2} \exp(-T a),$$
and therefore, $f(a)$ is approximately proportional to $\mathrm{Gamma}(a\mid a_0+n/2,\, b_0+T)$
when $a$ is large.  This suggests initializing with $A = a_0 + n/2$ and $B = b_0 + T$, as in Algorithm~\ref{algorithm:main}.

As $a \to 0$, we have $\Gamma(a)/a^a\sim a^{-1}$, since $a^a \to 1$ and $\Gamma(a)/a^{-1} = a\Gamma(a) = \Gamma(a + 1) \to \Gamma(1) = 1$ as $a\to 0$.
Plugging this into Equation~\ref{eqn:f}, we see that
$$p(x_{1:n} \mid a,\mu,a_0,b_0) \approx 
a^n \exp\!\big(-R - (T+n)a \big) \propto a^n \exp\!\big(-(T+n) a\big),$$
and therefore, $f(a)$ is approximately proportional to $\mathrm{Gamma}(a\mid a_0+n,\, b_0+T+n)$
when $a$ is small. This suggests initializing with $A = a_0 + n$ and $B = b_0 + T+n$.
While similar to the large $a$ version, 
empirically we find that this initialization tends to converge slightly slower, by 1 iteration or so.

\section{Numerical integration for truth comparison}
\label{section:integration}

Since the true full conditional distribution is not analytically tractable, it is necessary to approximate the integrals for the total variation distance and the Kullback--Leibler divergences between the true and approximate distributions. To do this, we use a simple adaptive numerical integration technique using an importance distribution to choose points (basically, a deterministic importance sampler).

Specifically, suppose we want to approximate the integral $\int h(a) g(a) d a$ where $g(a)$ is a probability density on $\R$, and $h(a)$ is some function on $\R$. Let $G(a)$ be the CDF of $g$, and let $a_i = G^{-1}(u_i)$ where $u_i = (i-0.5)/N$ for $i = 1,\ldots,N$. (If the inverse of $G$ does not exist, use the generalized inverse distribution function.)  Then
$\int h(a) g(a) d a \approx \frac{1}{N} \sum_{i = 1}^N h(a_i) $
when $N$ is large. This is similar to Monte Carlo, but is much more accurate due to using equally-spaced points $u_i$ rather than random draws $u_i \sim \mathrm{Uniform}(0,1)$.

We apply this approximation with $g(a) = \mathrm{Gamma}(a|A,B)$, where $A,B>0$ are the output from Algorithm~\ref{algorithm:main}. Let $f(a) = p(a \mid x_{1:n},\mu,a_0,b_0)$ be the true full conditional. Since the normalizing constant of $f$ is unknown, write $f(a) = \tilde f(a)/Z$ where $\tilde f(a)$ can be easily computed, and observe that
$$Z = \int_{-\infty}^\infty \tilde f(a) d a = \int_{0}^\infty \frac{\tilde f(a)}{g(a)} g(a) d a = \int_{-\infty}^\infty w(a) g(a) d a $$
where $w(a) = \tilde f(a)/g(a)$ for $a>0$ and $w(a) = 0$ for $a\leq 0$, since $g(a) > 0$ when $a>0$, and $f(a)=g(a)=0$ when $a\leq 0$.
Thus, $Z \approx \frac{1}{N}\sum_{i=1}^N w(a_i)$.

Define $\widehat f(a) = \tilde f(a) / \big(\frac{1}{N}\sum_{i=1}^N w(a_i)\big)$, so that $\widehat f(a)\approx f(a)$.
The integral for total variation distance can be written in the form $\int h(a) g(a) d a$ as follows:
$$ d_\mathrm{TV}(f,g) = \frac{1}{2} \int_{-\infty}^\infty |f(a)-g(a)|d a = \frac{1}{2} \int_0^\infty \Big|\frac{f(a)}{g(a)}-1\Big|g(a) d a = \int_{-\infty}^\infty h(a) g(a) d a $$
where $h(a) = \frac{1}{2}|f(a)/g(a)-1|$ for $a>0$ and $h(a) = 0$ for $a\leq 0$. Thus, since $a_i>0$,
$$d_\mathrm{TV}(f,g) \approx \frac{1}{N}\sum_{i=1}^N \frac{1}{2}\Big|\frac{\widehat f(a_i)}{g(a_i)} - 1\Big|.$$
Similarly, 
$\displaystyle d_\mathrm{KL}(f,g) \approx \frac{1}{N}\sum_{i=1}^N \frac{\widehat f(a_i)}{g(a_i)}\log\frac{\widehat f(a_i)}{g(a_i)} $
and 
$\displaystyle d_\mathrm{KL}(g,f) \approx \frac{1}{N}\sum_{i=1}^N \log\frac{g(a_i)}{\widehat f(a_i)}. $

\section{Additional empirical results}
\label{section:additional-results}

Figures~\ref{figure:TV-1}, \ref{figure:TV-2}, and \ref{figure:TV-3} show the total variation distance for all the cases with $a_0 = 1$, $a_0 = 0.1$, and $a_0 = 0.01$, respectively.
Figures~\ref{figure:KL2-1} and \ref{figure:KL1-1} show the Kullback--Leibler divergences for all the cases with $a_0 = 1$.

\begin{figure}
  \centering
  \includegraphics[trim=.3cm 0 1.2cm 0, clip, width=0.325\textwidth]{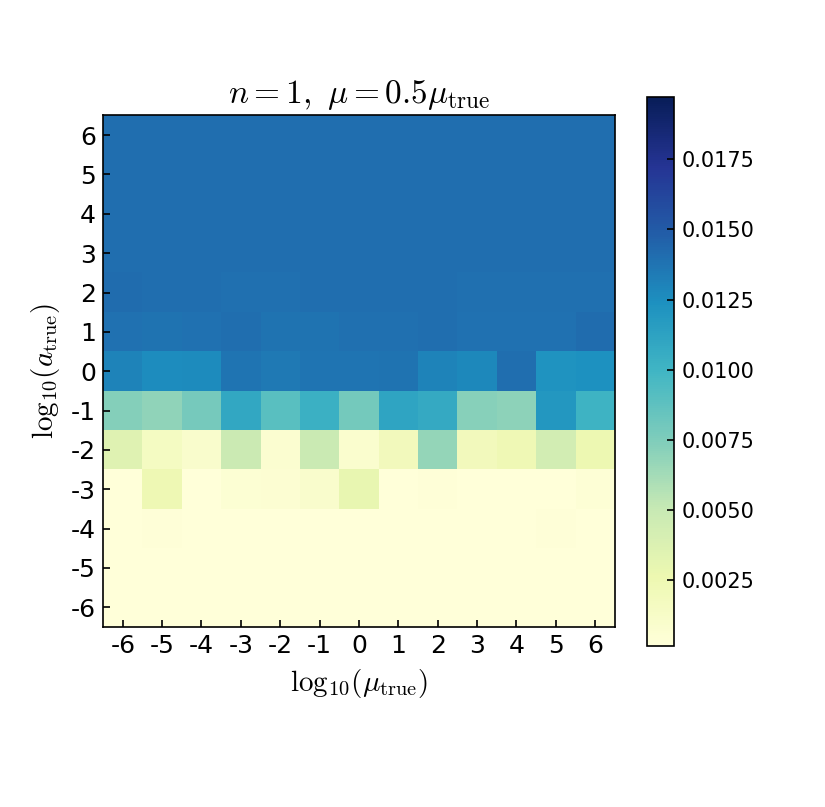}
  \includegraphics[trim=.3cm 0 1.2cm 0, clip, width=0.325\textwidth]{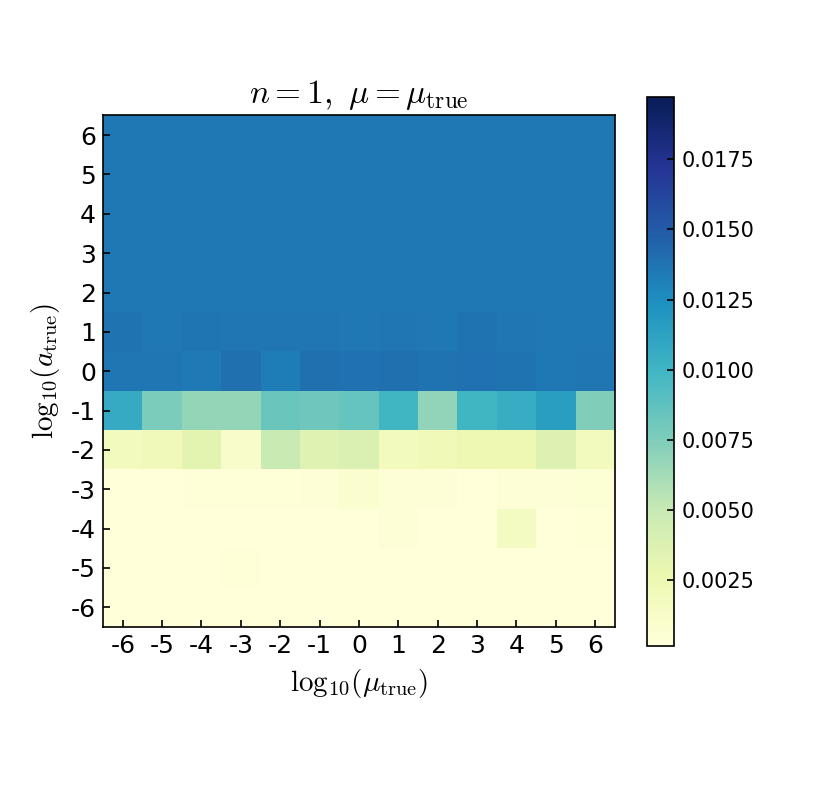}
  \includegraphics[trim=.3cm 0 1.2cm 0, clip, width=0.325\textwidth]{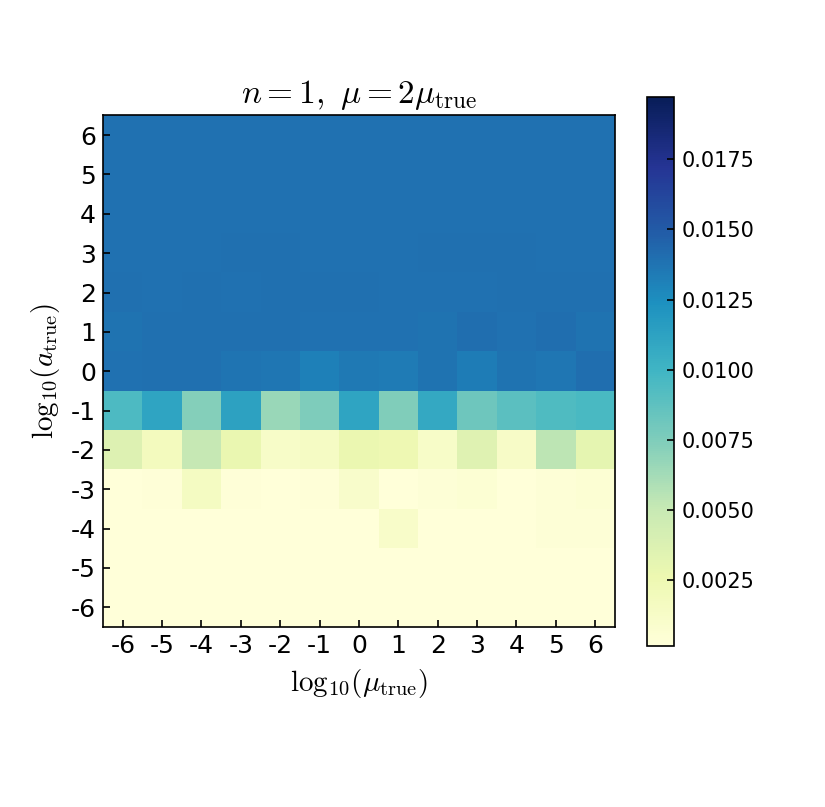}
  \includegraphics[trim=.3cm 0 1.2cm 0, clip, width=0.325\textwidth]{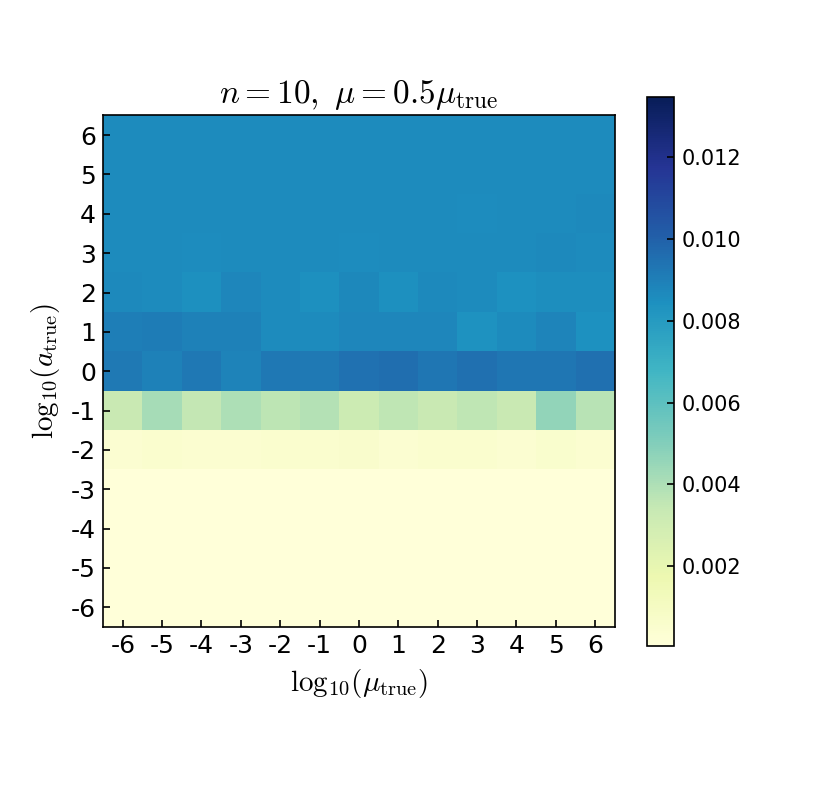}
  \includegraphics[trim=.3cm 0 1.2cm 0, clip, width=0.325\textwidth]{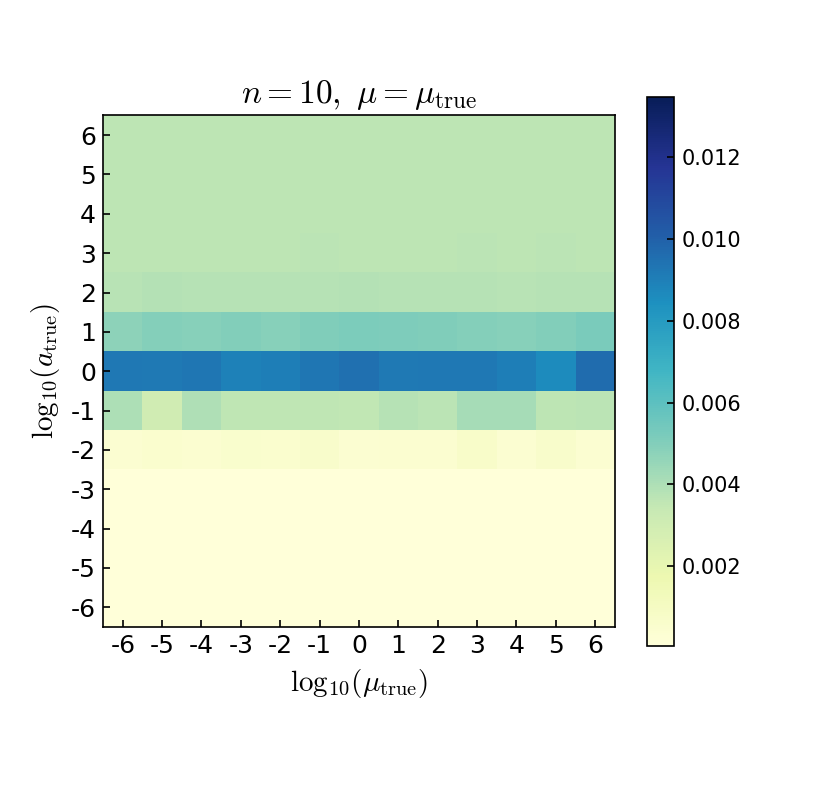}
  \includegraphics[trim=.3cm 0 1.2cm 0, clip, width=0.325\textwidth]{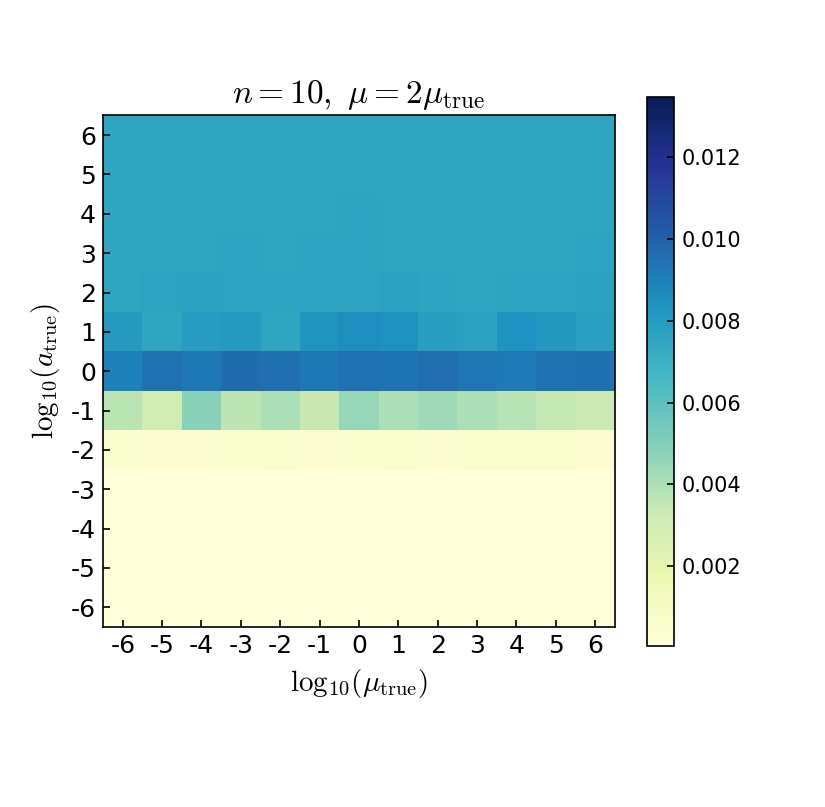}
  \includegraphics[trim=.3cm 0 1.2cm 0, clip, width=0.325\textwidth]{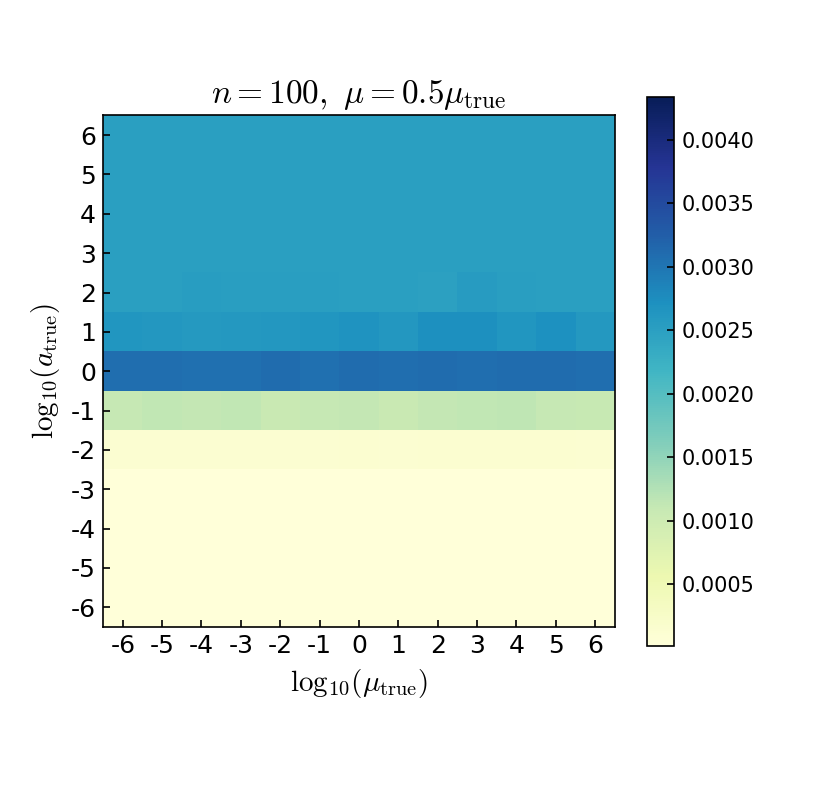}
  \includegraphics[trim=.3cm 0 1.2cm 0, clip, width=0.325\textwidth]{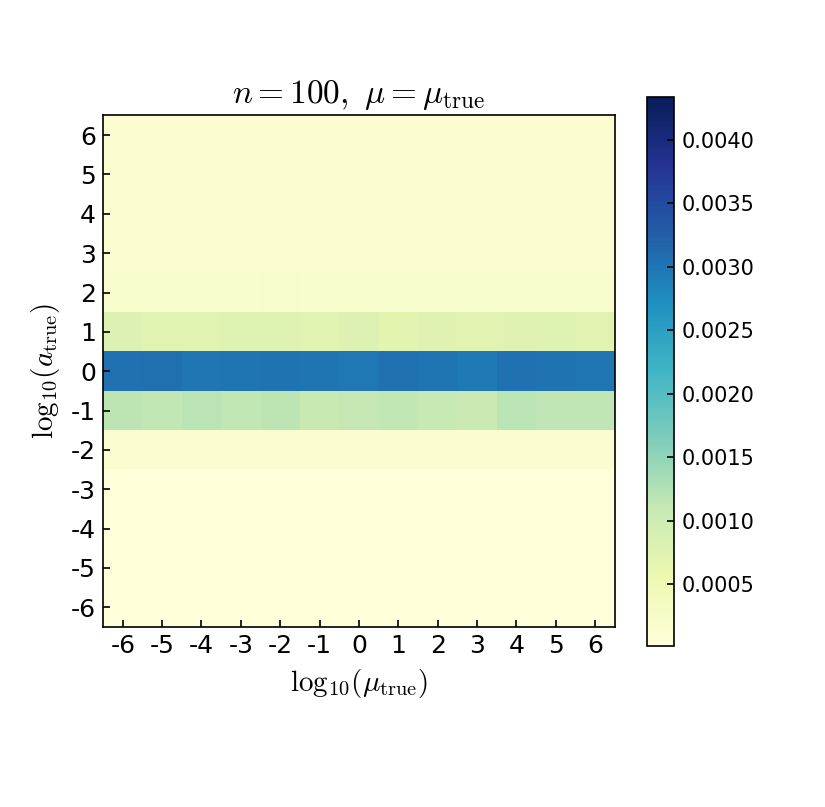}
  \includegraphics[trim=.3cm 0 1.2cm 0, clip, width=0.325\textwidth]{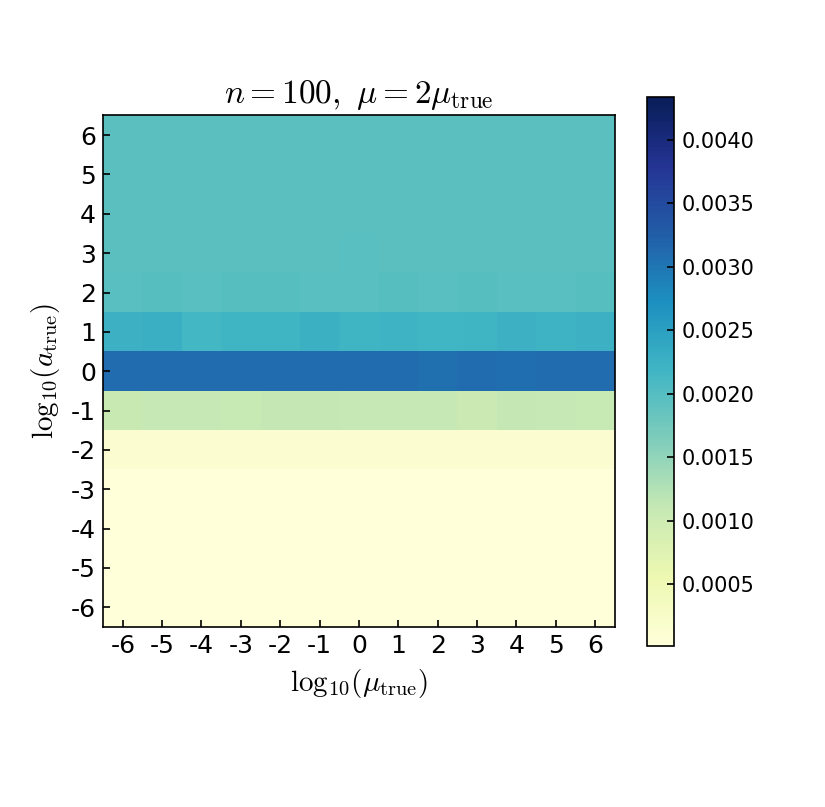}
  \caption{Total variation distance $d_\mathrm{TV}(f,g)$ between the true full conditional $f$ and the approximate full conditional $g$, for every case with $a_0 = 1$. The values shown are the averages over the five independent data sets for each case. Note the scale at the right of each plot.}
  \label{figure:TV-1}
\end{figure}

\begin{figure}
  \centering
  \includegraphics[trim=.3cm 0 1.2cm 0, clip, width=0.325\textwidth]{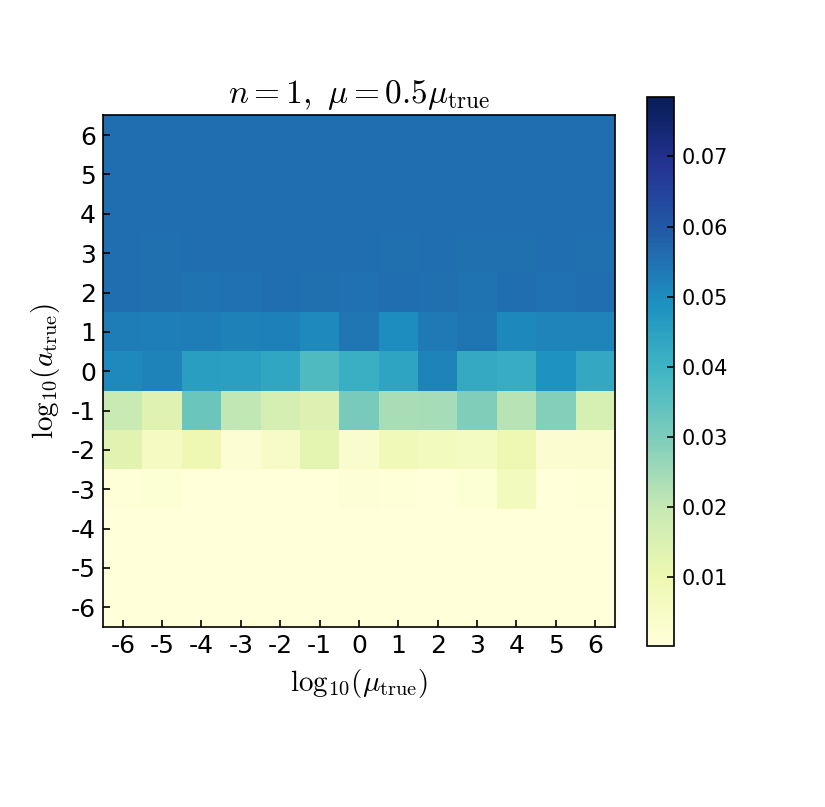}
  \includegraphics[trim=.3cm 0 1.2cm 0, clip, width=0.325\textwidth]{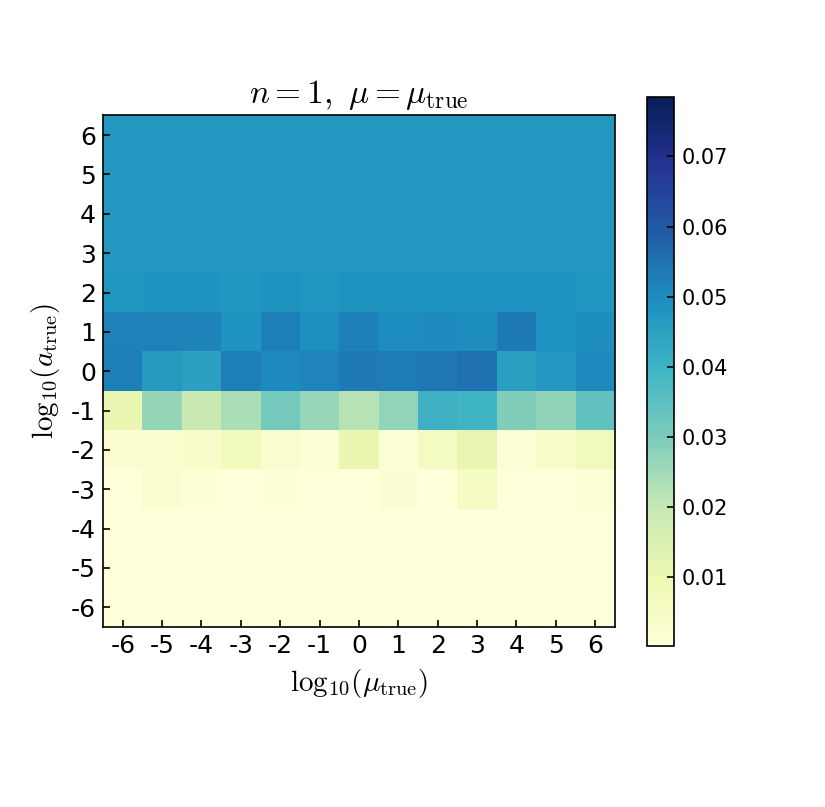}
  \includegraphics[trim=.3cm 0 1.2cm 0, clip, width=0.325\textwidth]{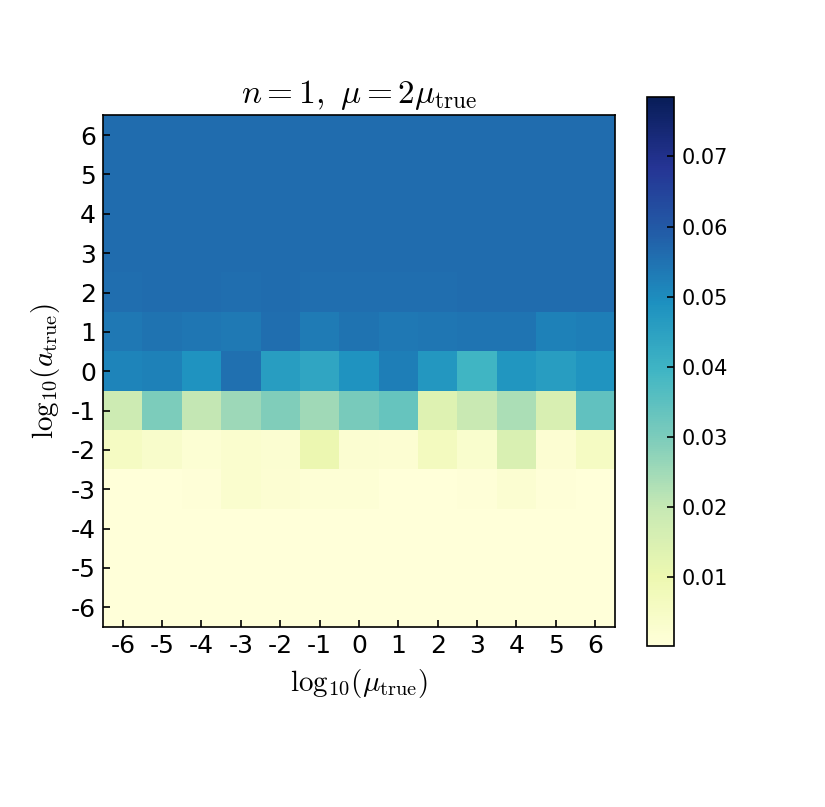}
  \includegraphics[trim=.3cm 0 1.2cm 0, clip, width=0.325\textwidth]{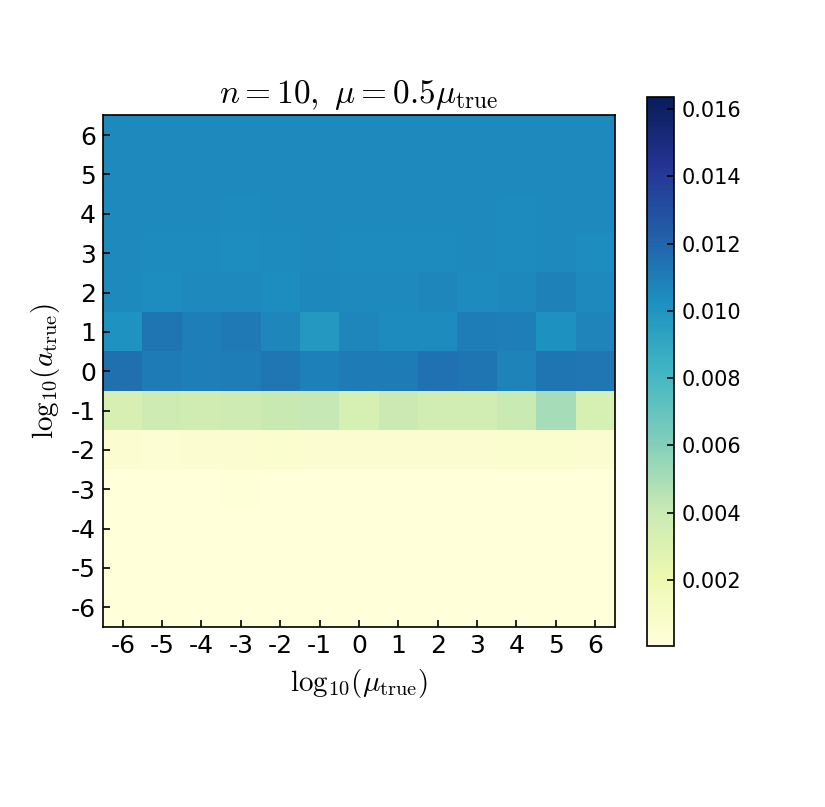}
  \includegraphics[trim=.3cm 0 1.2cm 0, clip, width=0.325\textwidth]{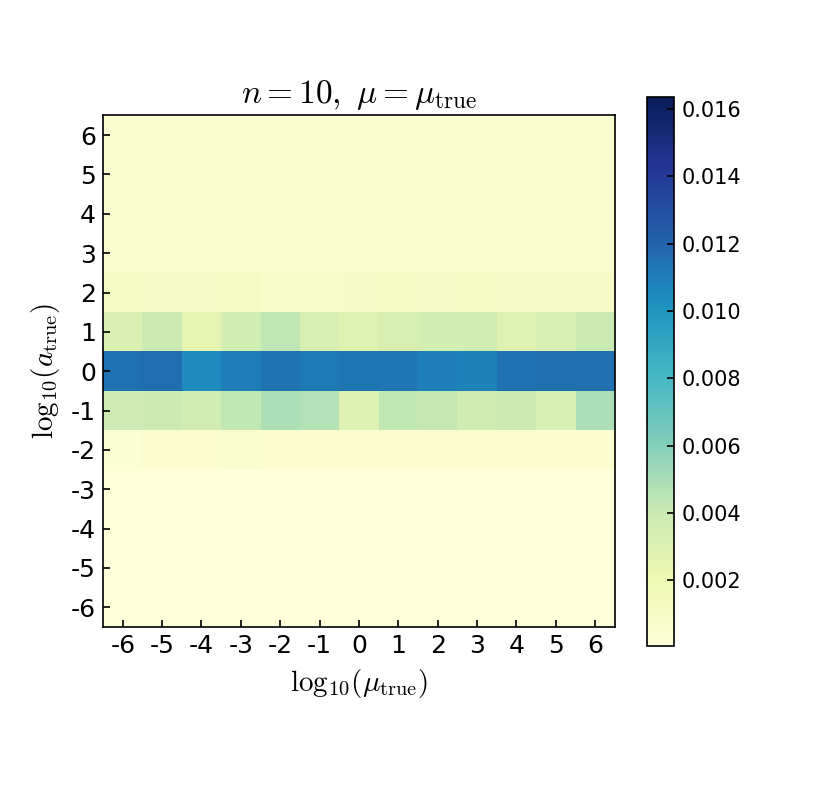}
  \includegraphics[trim=.3cm 0 1.2cm 0, clip, width=0.325\textwidth]{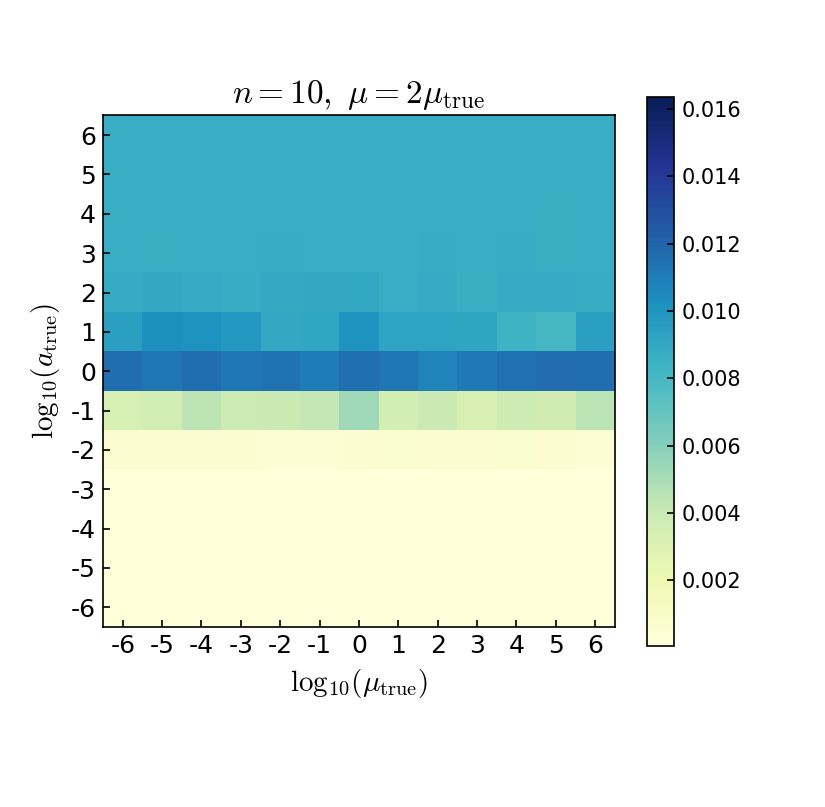}
  \includegraphics[trim=.3cm 0 1.2cm 0, clip, width=0.325\textwidth]{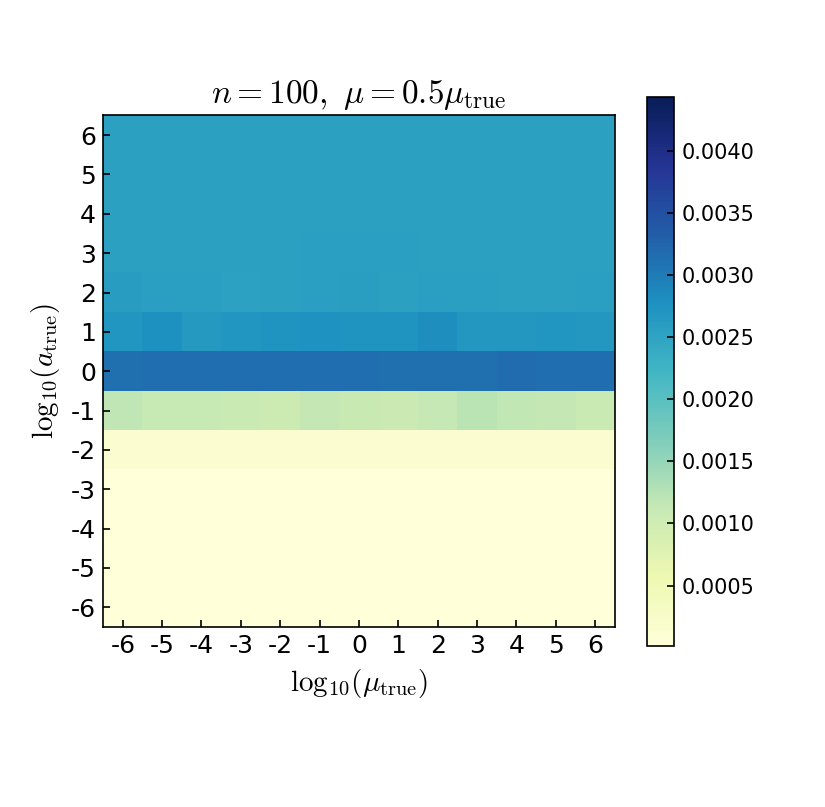}
  \includegraphics[trim=.3cm 0 1.2cm 0, clip, width=0.325\textwidth]{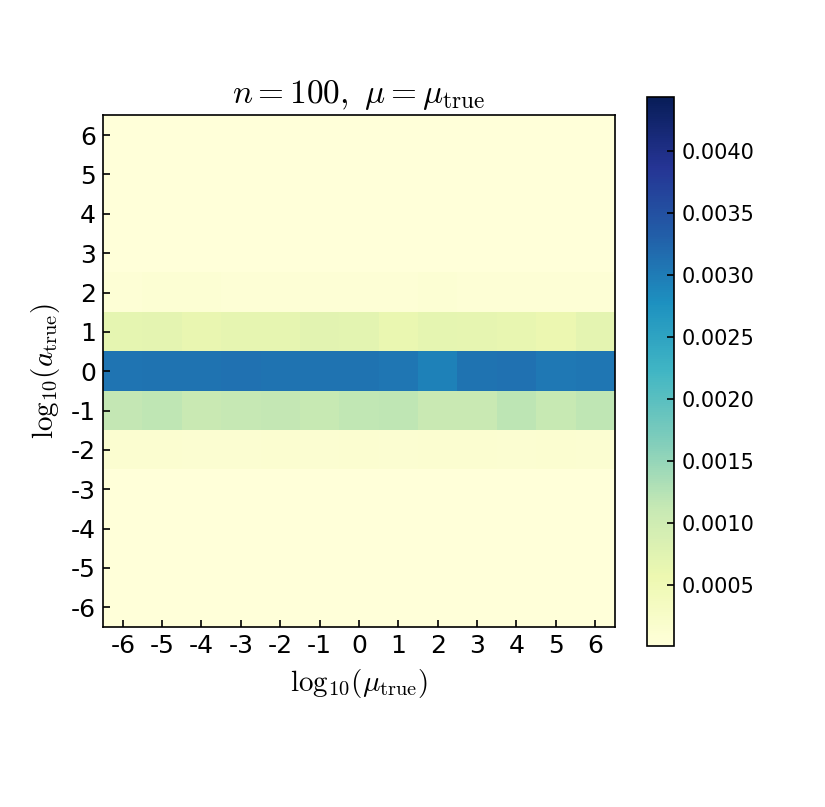}
  \includegraphics[trim=.3cm 0 1.2cm 0, clip, width=0.325\textwidth]{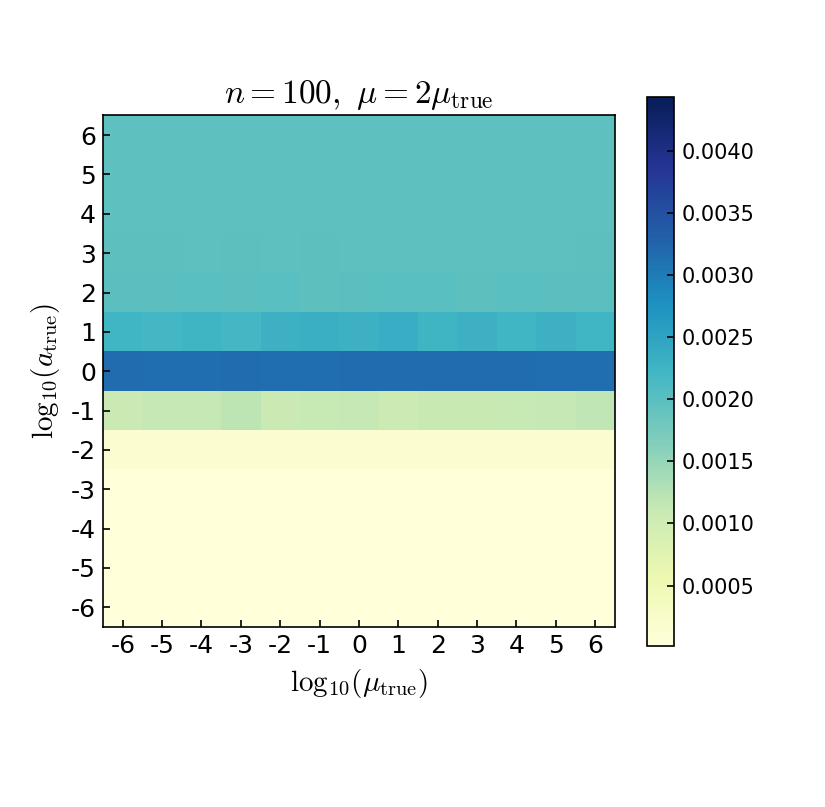}
  \caption{Total variation distance $d_\mathrm{TV}(f,g)$ between the true full conditional $f$ and the approximate full conditional $g$, for every case with $a_0 = 0.1$. The values shown are the averages over the five independent data sets for each case. Note the scale at the right of each plot.}
  \label{figure:TV-2}
\end{figure}

\begin{figure}
  \centering
  \includegraphics[trim=.3cm 0 1.2cm 0, clip, width=0.325\textwidth]{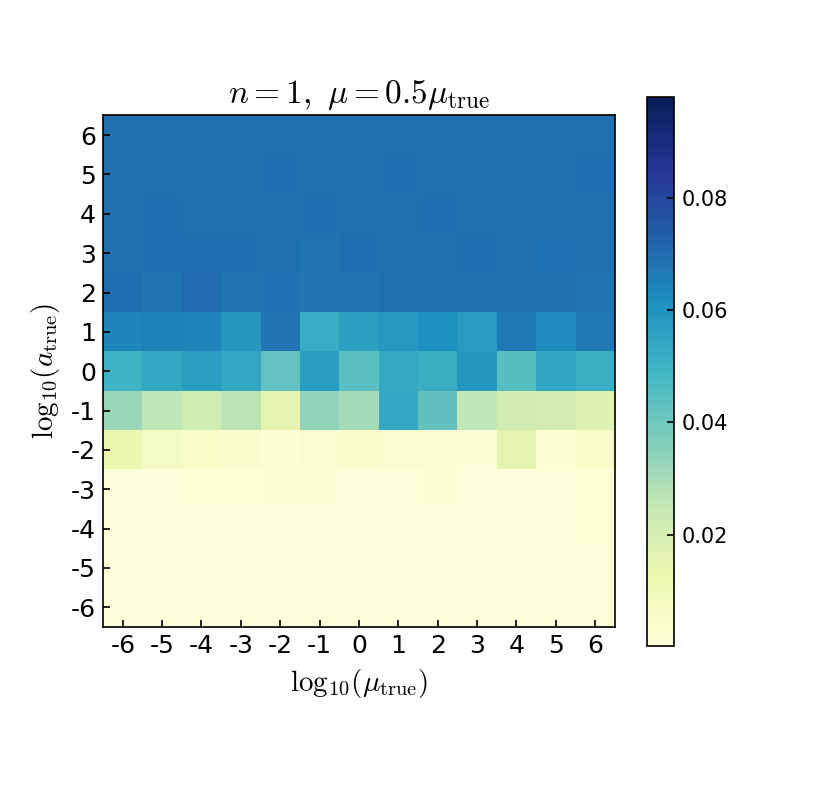}
  \includegraphics[trim=.3cm 0 1.2cm 0, clip, width=0.325\textwidth]{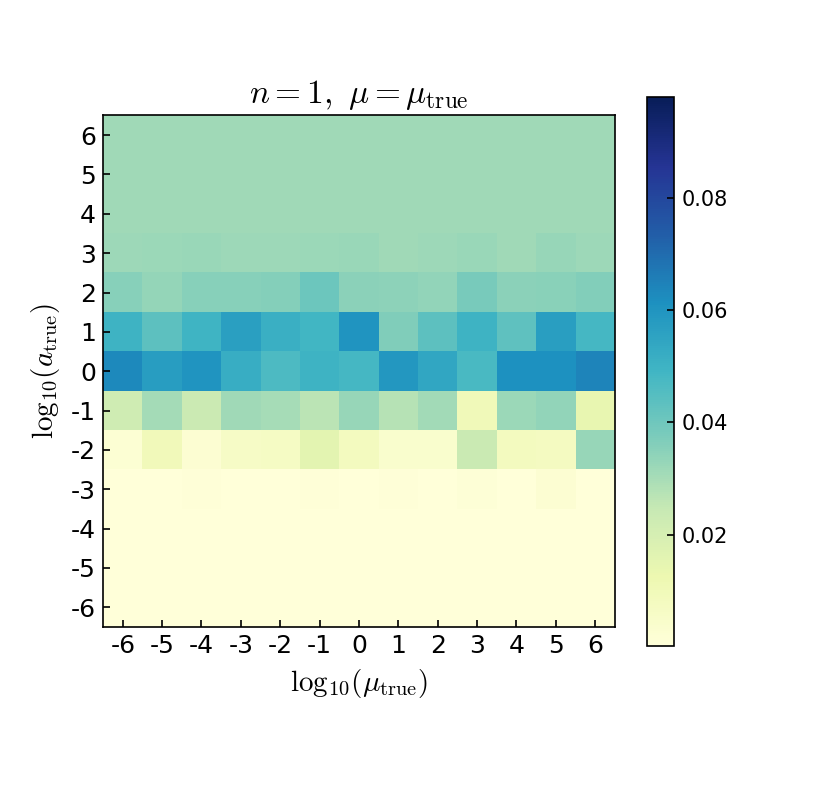}
  \includegraphics[trim=.3cm 0 1.2cm 0, clip, width=0.325\textwidth]{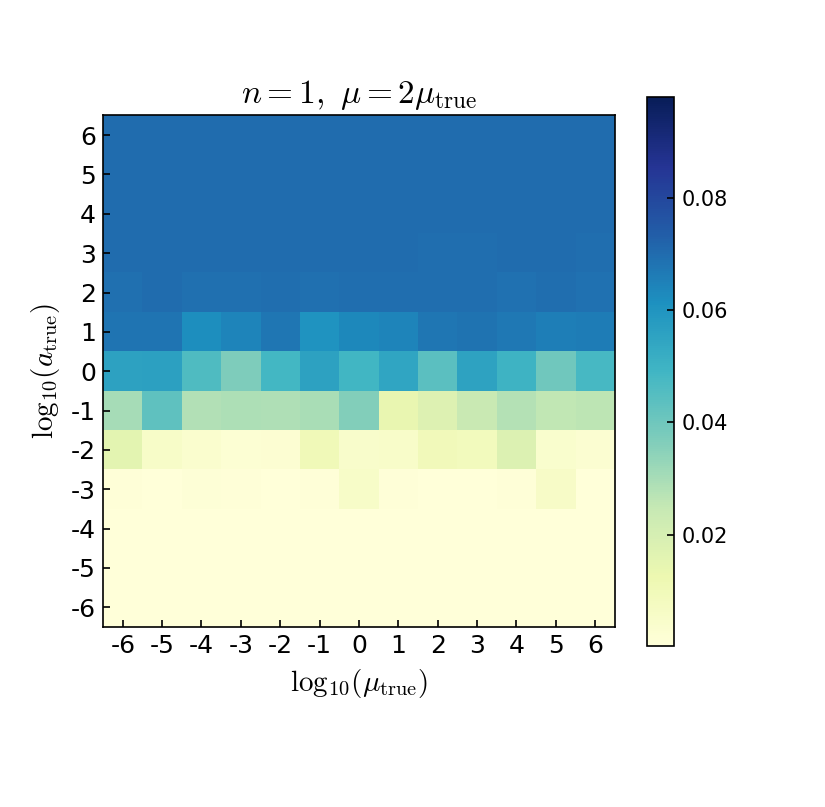}
  \includegraphics[trim=.3cm 0 1.2cm 0, clip, width=0.325\textwidth]{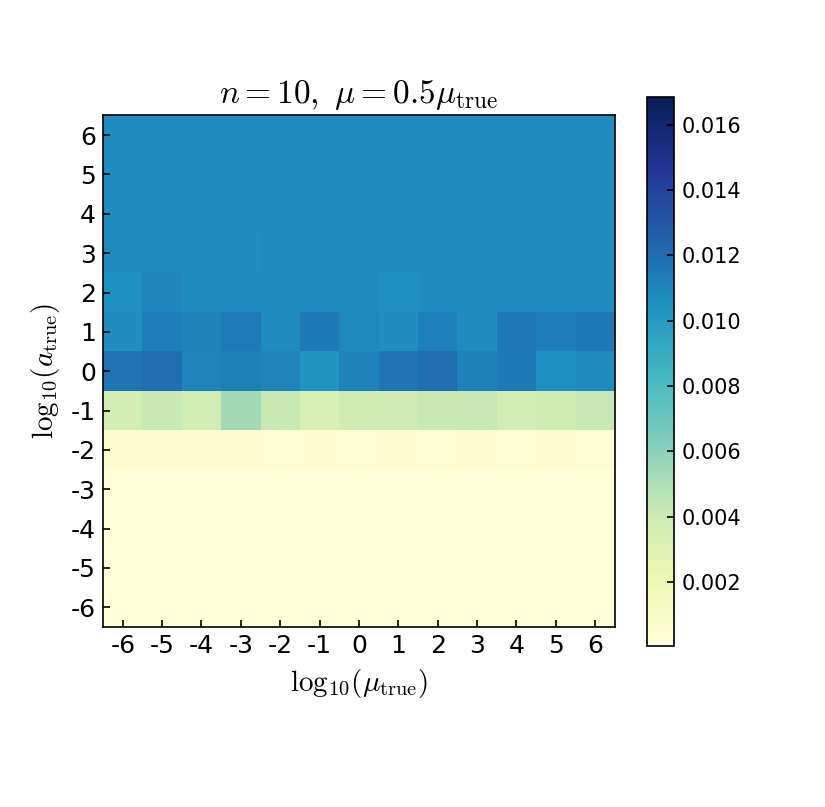}
  \includegraphics[trim=.3cm 0 1.2cm 0, clip, width=0.325\textwidth]{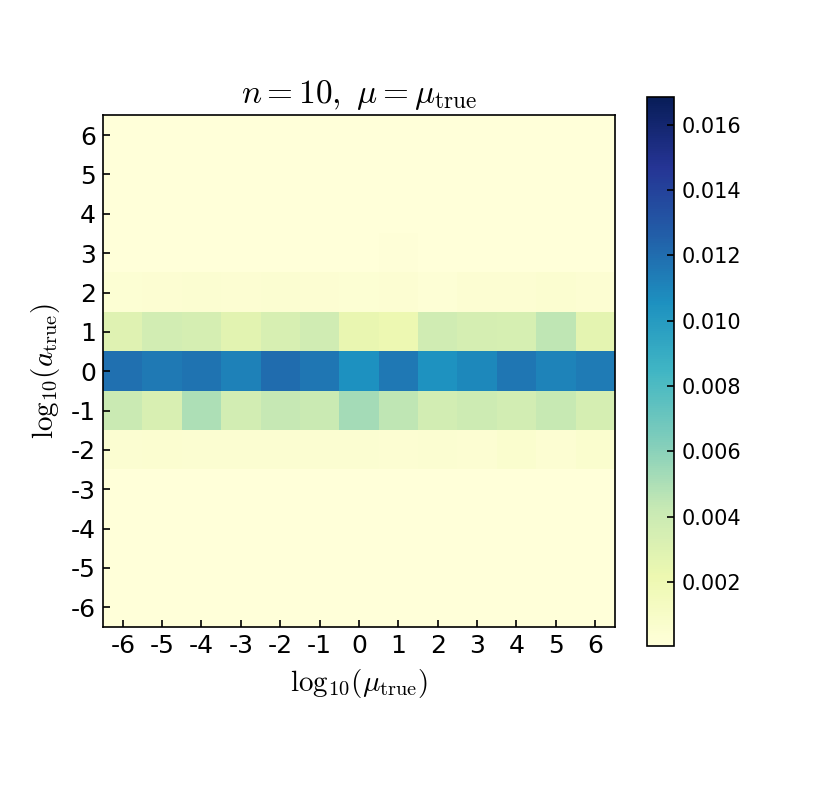}
  \includegraphics[trim=.3cm 0 1.2cm 0, clip, width=0.325\textwidth]{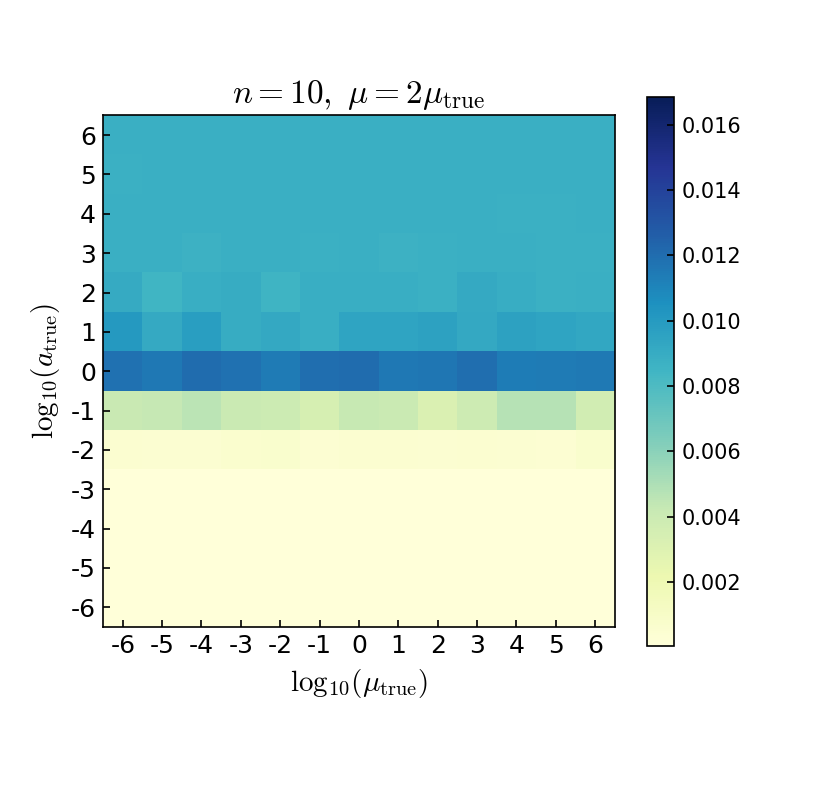}
  \includegraphics[trim=.3cm 0 1.2cm 0, clip, width=0.325\textwidth]{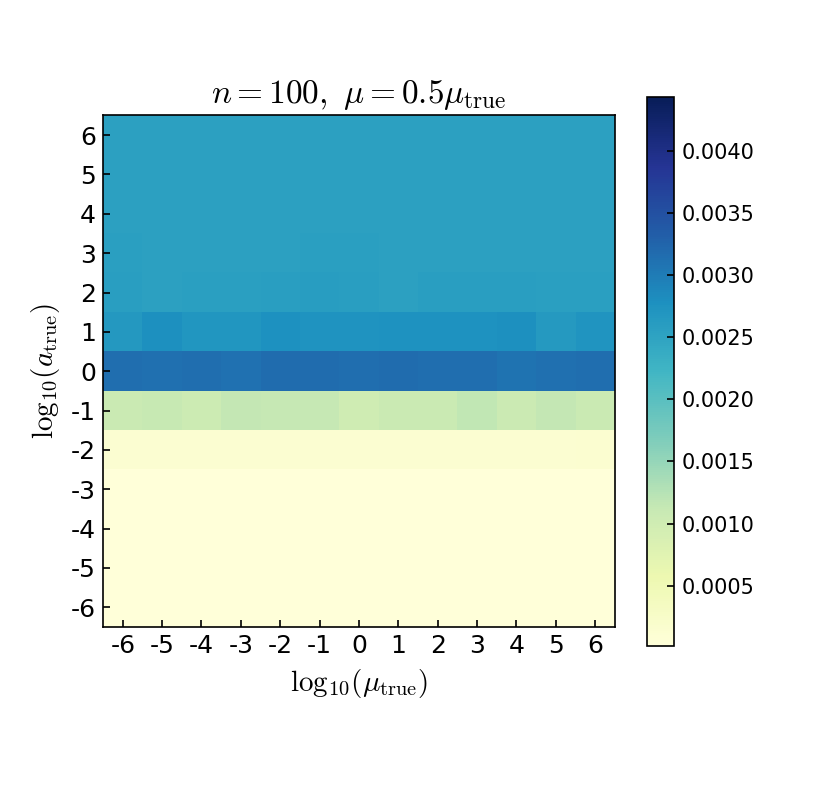}
  \includegraphics[trim=.3cm 0 1.2cm 0, clip, width=0.325\textwidth]{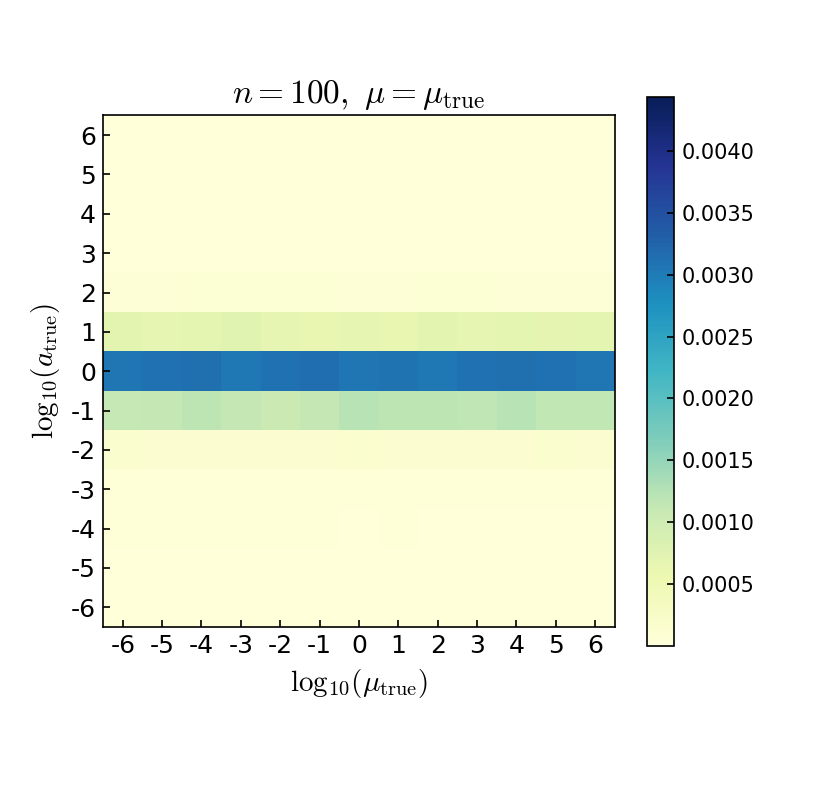}
  \includegraphics[trim=.3cm 0 1.2cm 0, clip, width=0.325\textwidth]{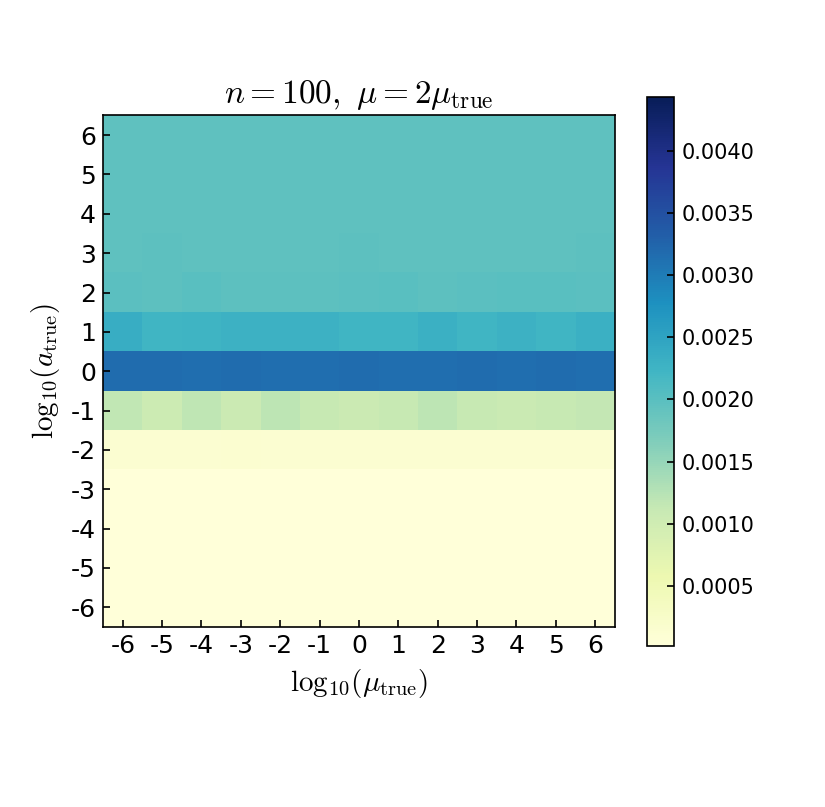}
  \caption{Total variation distance $d_\mathrm{TV}(f,g)$ between the true full conditional $f$ and the approximate full conditional $g$, for every case with $a_0 = 0.01$. The values shown are the averages over the five independent data sets for each case. Note the scale at the right of each plot.}
  \label{figure:TV-3}
\end{figure}

\begin{figure}
  \centering
  \includegraphics[trim=.3cm 0 1cm 0, clip, width=0.325\textwidth]{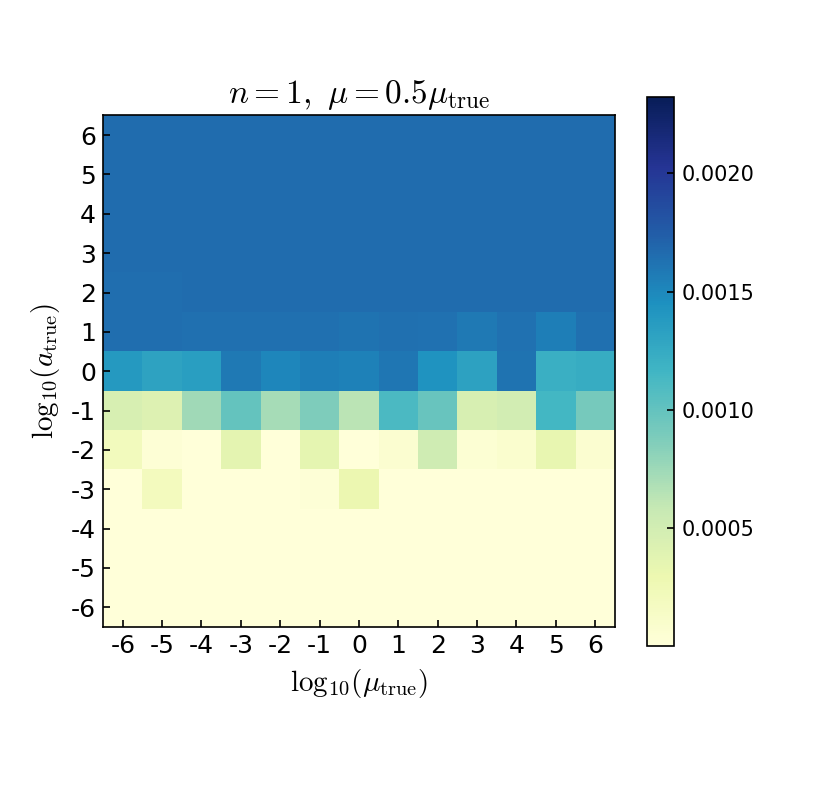}
  \includegraphics[trim=.3cm 0 1cm 0, clip, width=0.325\textwidth]{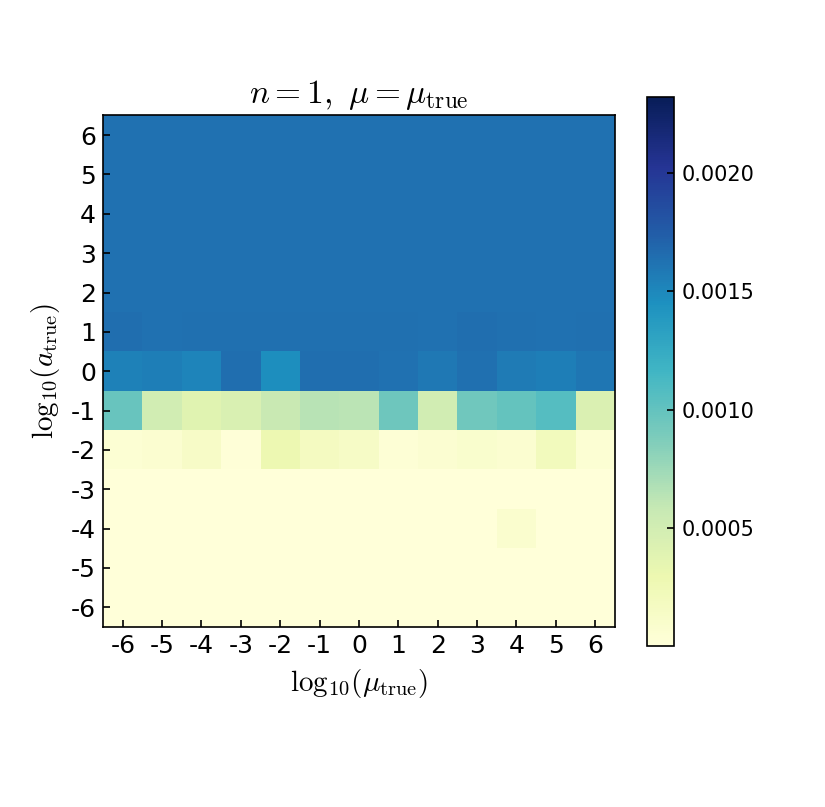}
  \includegraphics[trim=.3cm 0 1cm 0, clip, width=0.325\textwidth]{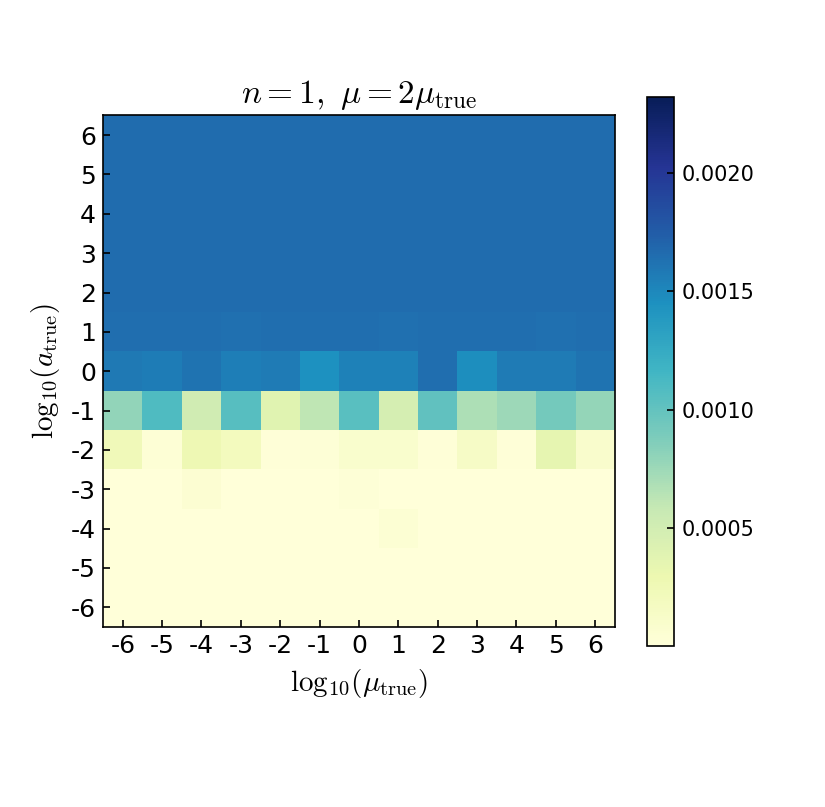}
  \includegraphics[trim=.3cm 0 1cm 0, clip, width=0.325\textwidth]{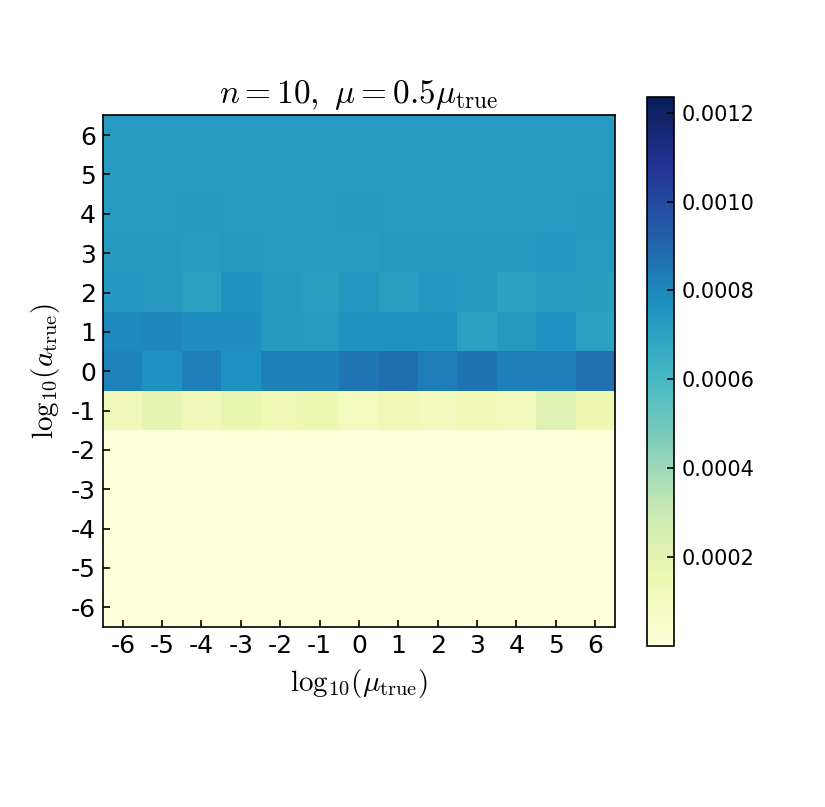}
  \includegraphics[trim=.3cm 0 1cm 0, clip, width=0.325\textwidth]{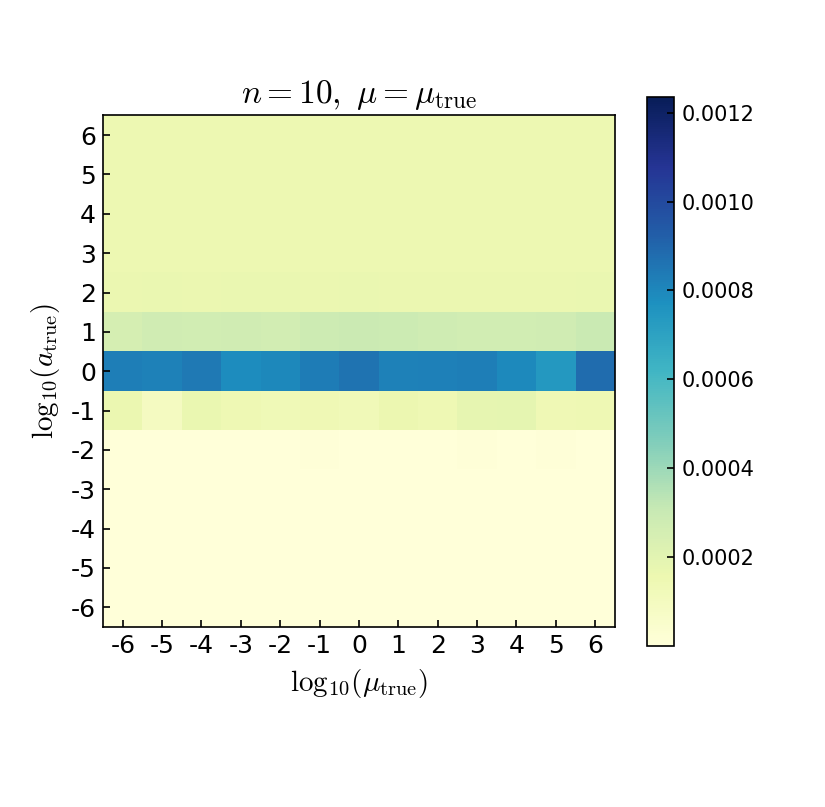}
  \includegraphics[trim=.3cm 0 1cm 0, clip, width=0.325\textwidth]{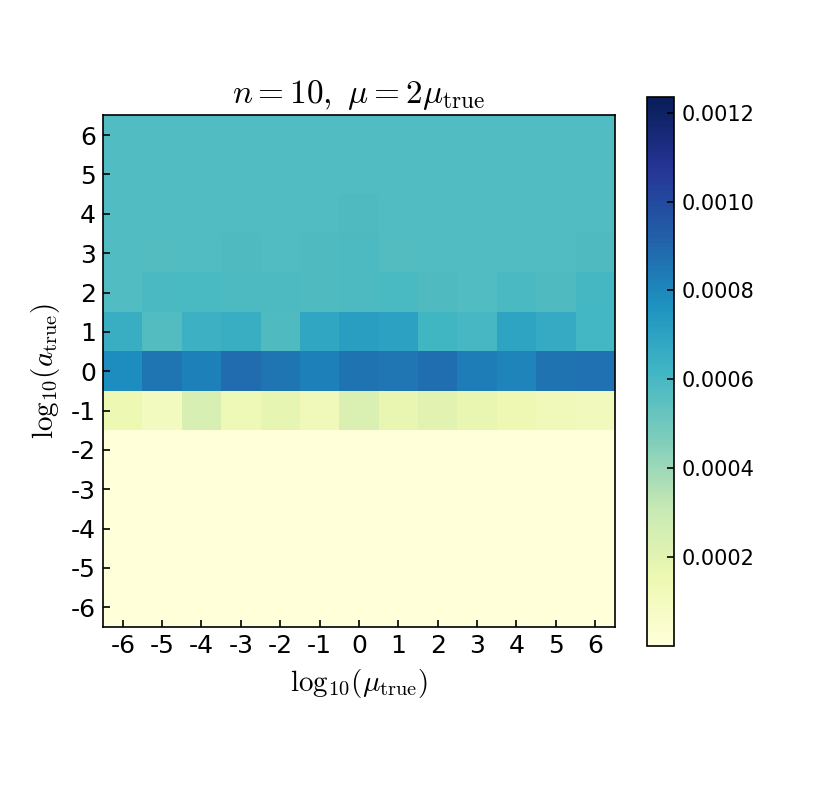}
  \includegraphics[trim=.3cm 0 1cm 0, clip, width=0.325\textwidth]{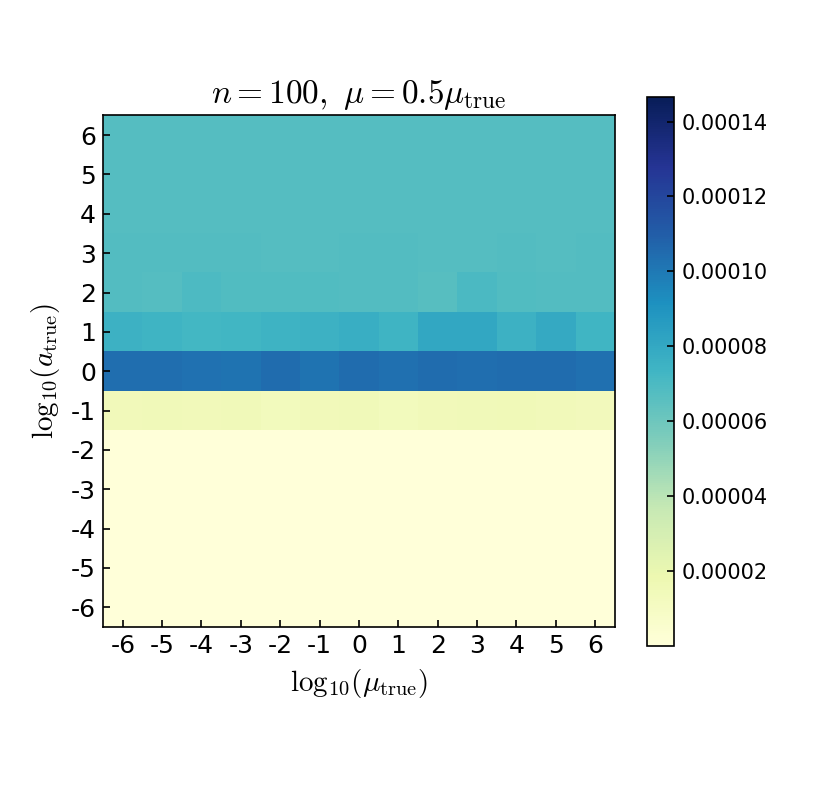}
  \includegraphics[trim=.3cm 0 1cm 0, clip, width=0.325\textwidth]{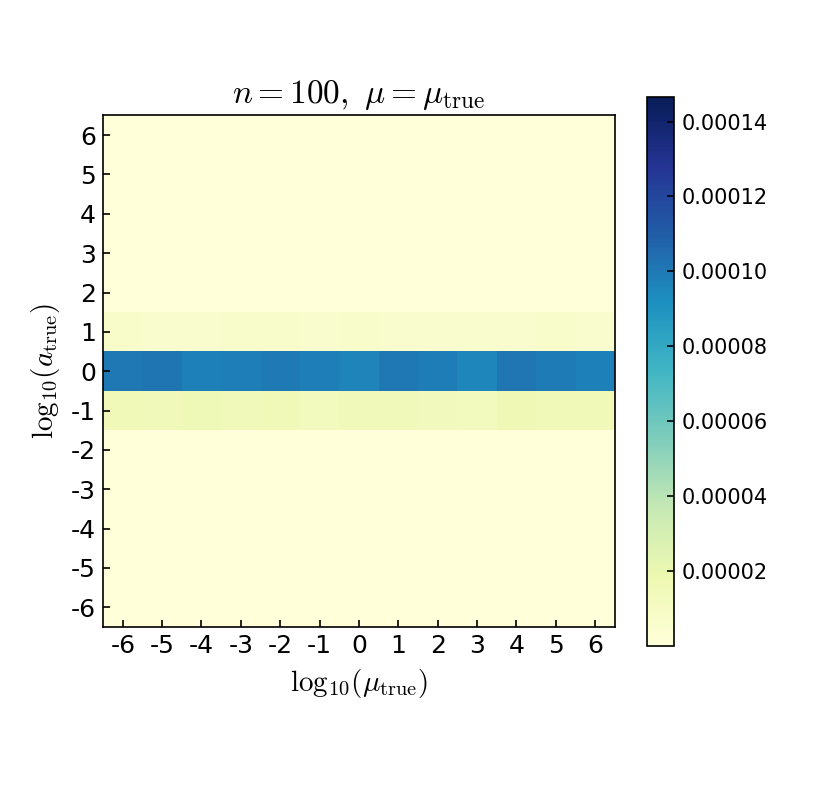}
  \includegraphics[trim=.3cm 0 1cm 0, clip, width=0.325\textwidth]{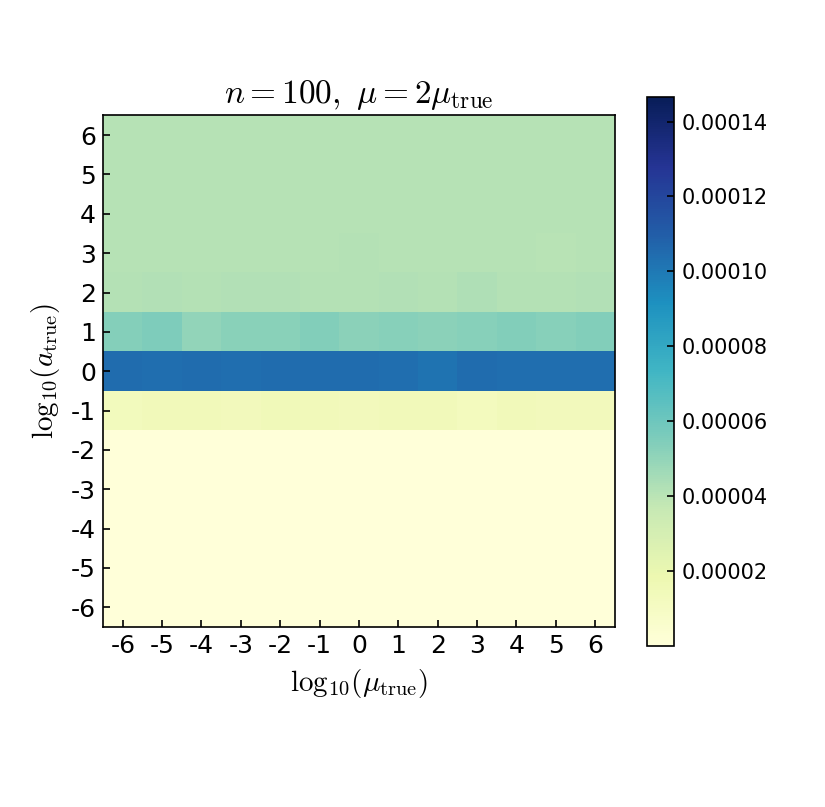}
  \caption{Kullback--Leibler divergence $d_\mathrm{KL}(f,g)$ between the true full conditional $f$ and the approximate full conditional $g$, for every case with $a_0 = 1$. The values shown are the averages over the five independent data sets for each case. Note the scale at the right of each plot.}
  \label{figure:KL2-1}
\end{figure}

\begin{figure}
  \centering
  \includegraphics[trim=.3cm 0 1cm 0, clip, width=0.325\textwidth]{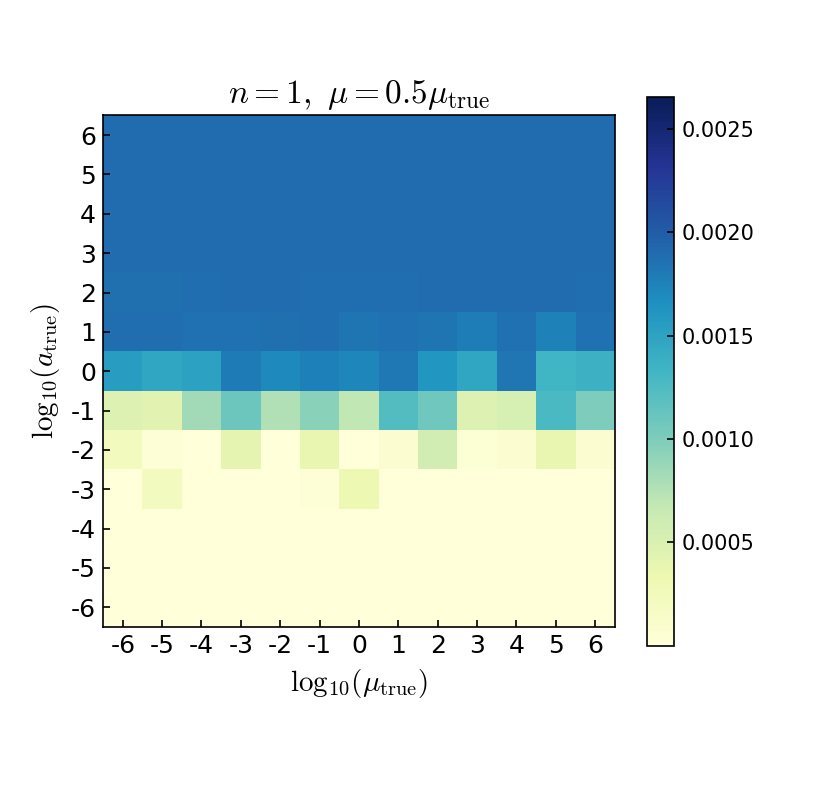}
  \includegraphics[trim=.3cm 0 1cm 0, clip, width=0.325\textwidth]{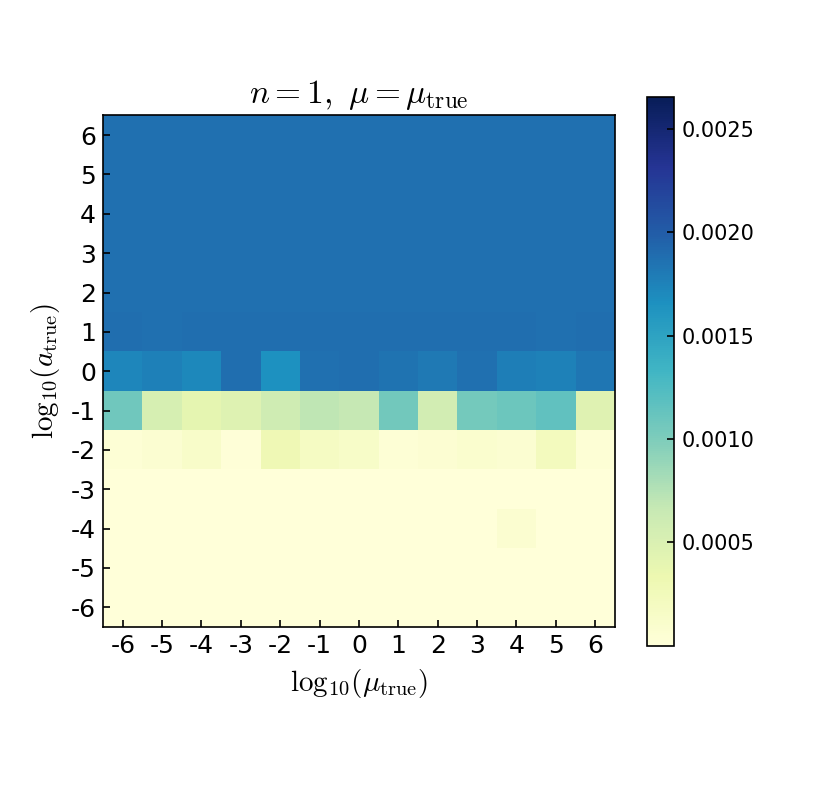}
  \includegraphics[trim=.3cm 0 1cm 0, clip, width=0.325\textwidth]{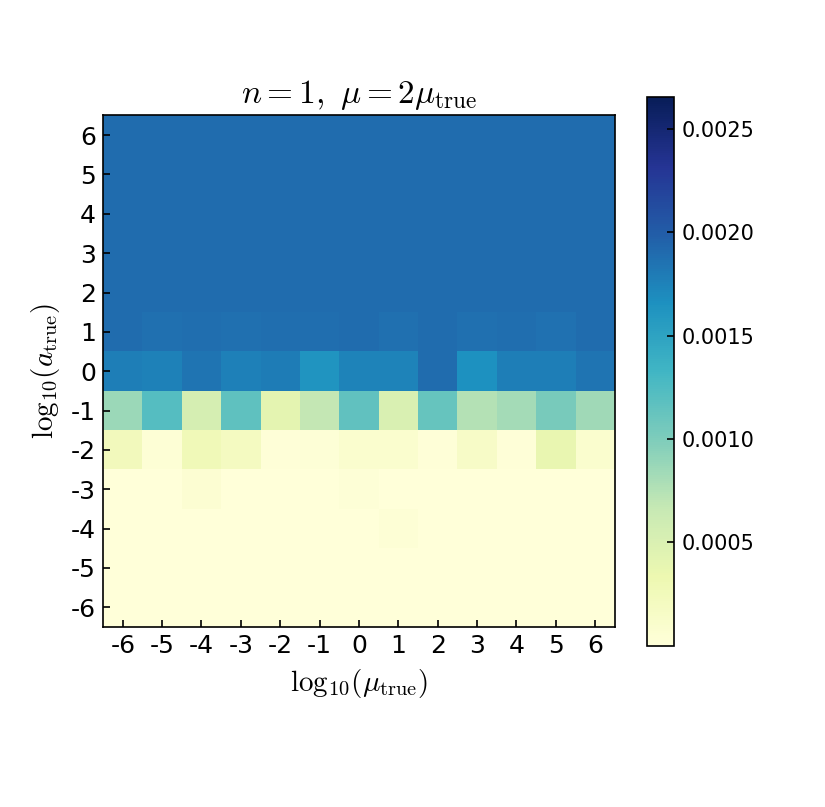}
  \includegraphics[trim=.3cm 0 1cm 0, clip, width=0.325\textwidth]{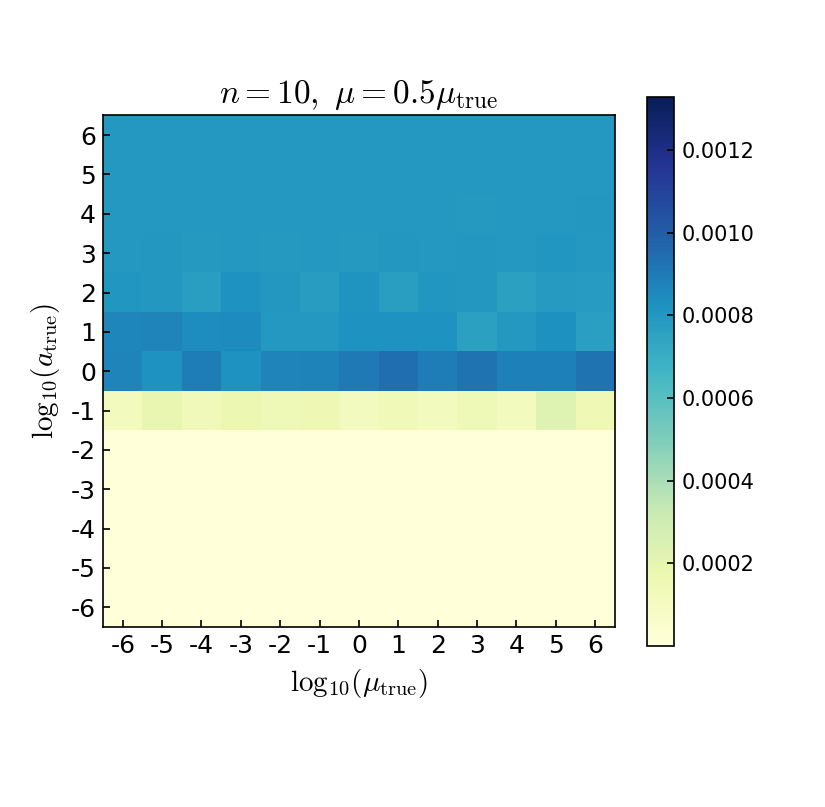}
  \includegraphics[trim=.3cm 0 1cm 0, clip, width=0.325\textwidth]{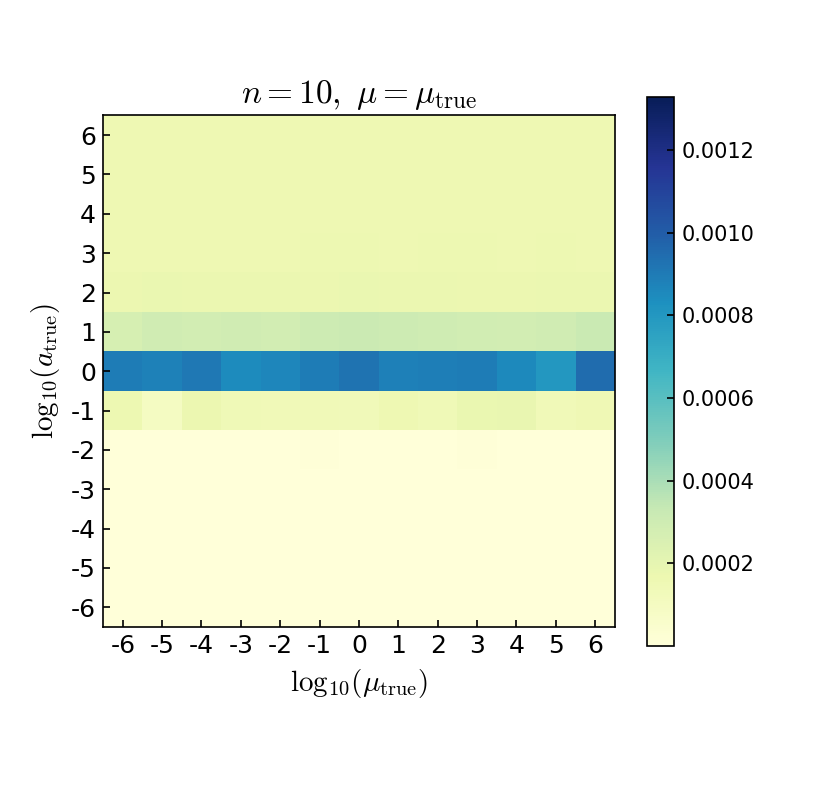}
  \includegraphics[trim=.3cm 0 1cm 0, clip, width=0.325\textwidth]{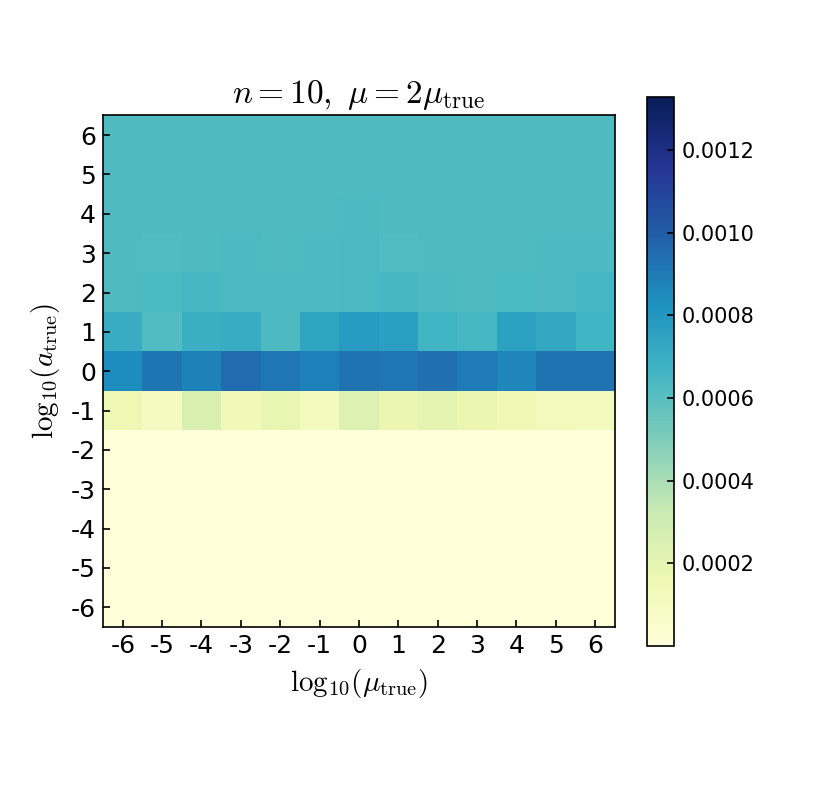}
  \includegraphics[trim=.3cm 0 1cm 0, clip, width=0.325\textwidth]{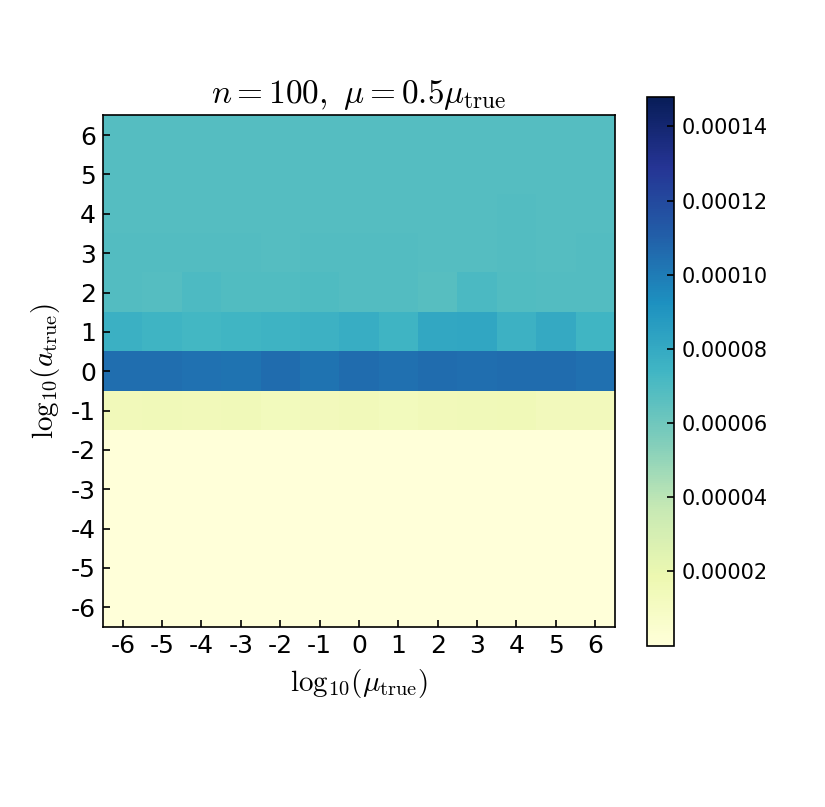}
  \includegraphics[trim=.3cm 0 1cm 0, clip, width=0.325\textwidth]{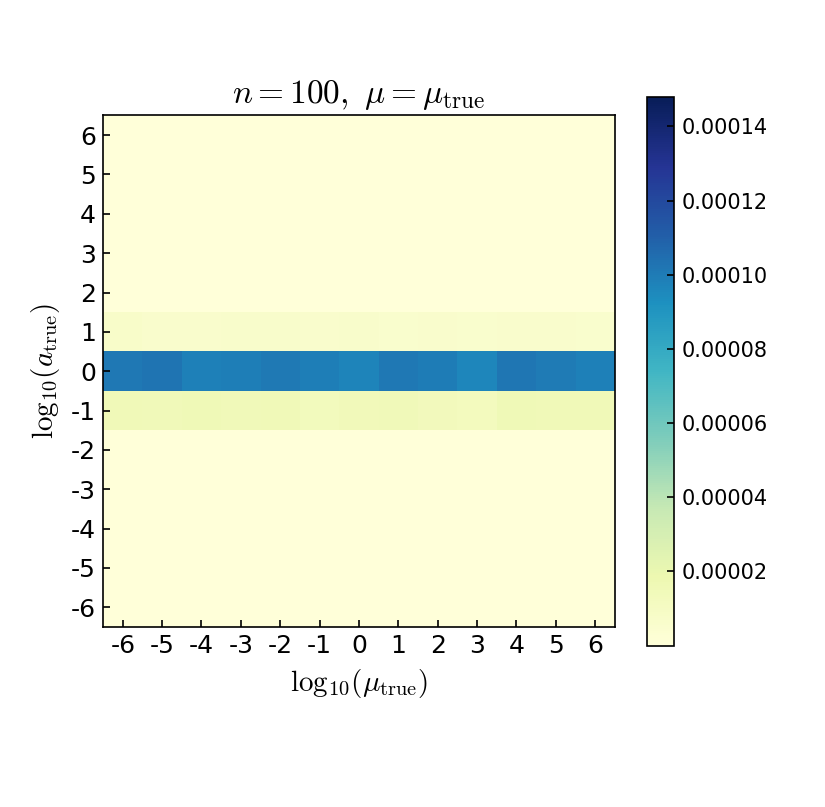}
  \includegraphics[trim=.3cm 0 1cm 0, clip, width=0.325\textwidth]{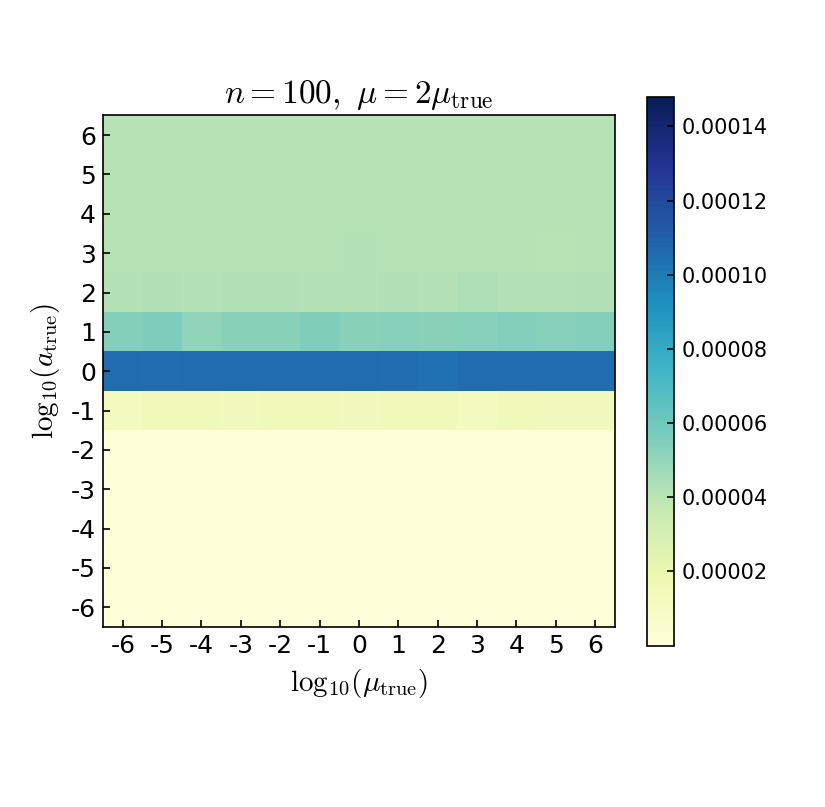}
  \caption{Kullback--Leibler divergence $d_\mathrm{KL}(g,f)$ between the true full conditional $f$ and the approximate full conditional $g$, for every case with $a_0 = 1$. The values shown are the averages over the five independent data sets for each case. Note the scale at the right of each plot.}
  \label{figure:KL1-1}
\end{figure}

\bibliographystyle{abbrvnat}
% \bibliography{supp}

\end{document}